\documentclass{aart}
\usepackage{amsmath,amsthm,amsfonts,latexsym, amssymb, mathrsfs}
\usepackage{amsxtra}
\usepackage{amscd}
\usepackage{color}

\usepackage{bbm}
\usepackage{mathrsfs}

\usepackage{xspace}

\newcommand{\hp}{with high probability\xspace}

\newcommand{\f}[1]{\boldsymbol{\mathrm{#1}}} 
 \newcommand{\rr}{\mathrm} 
\renewcommand{\cal}{\mathcal}

\newcommand{\ol}[1]{\overline{#1} \!\,} 

\newcommand{\wt}{\widetilde}

\newcommand{\me}{\mathrm{e}}
\newcommand{\ii}{\mathrm{i}}
\newcommand{\dd}{\mathrm{d}}
\newcommand{\col}{\mathrel{\mathop:}}
\newcommand{\st}{\,\col\,}
\newcommand{\deq}{\mathrel{\mathop:}=}

\renewcommand{\leq}{\leqslant}
\renewcommand{\geq}{\geqslant}
\renewcommand{\le}{\leqslant}
\renewcommand{\ge}{\geqslant}

\newcommand{\for}{\qquad \text{for} \quad}

\newcommand{\ind}[1]{{\bf 1} (#1)}
\newcommand{\indb}[1]{{\bf 1} \pb{#1}}
\newcommand{\indB}[1]{{\bf 1} \pB{#1}}

\renewcommand{\epsilon}{\varepsilon}

\newcommand{\al}{\alpha}

\renewcommand{\P}{\mathbb{P}}
\newcommand{\E}{\mathbb{E}}
\newcommand{\R}{\mathbb{R}}

\newcommand{\N}{\mathbb{N}}

\newcommand{\pb}[1]{\bigl({#1}\bigr)}
\newcommand{\pB}[1]{\Bigl({#1}\Bigr)}
\newcommand{\pbb}[1]{\biggl({#1}\biggr)}
\newcommand{\pBB}[1]{\Biggl({#1}\Biggr)}

\newcommand{\qb}[1]{\bigl[{#1}\bigr]}
\newcommand{\qB}[1]{\Bigl[{#1}\Bigr]}
\newcommand{\qbb}[1]{\biggl[{#1}\biggr]}
\newcommand{\qBB}[1]{\Biggl[{#1}\Biggr]}

\newcommand{\h}[1]{\{{#1}\}}

\newcommand{\hB}[1]{\Bigl\{{#1}\Bigr\}}
\newcommand{\hbb}[1]{\biggl\{{#1}\biggr\}}
\newcommand{\hBB}[1]{\Biggl\{{#1}\Biggr\}}

\newcommand{\abs}[1]{\lvert #1 \rvert}
\newcommand{\absb}[1]{\bigl\lvert #1 \bigr\rvert}
\newcommand{\absB}[1]{\Bigl\lvert #1 \Bigr\rvert}
\newcommand{\absbb}[1]{\biggl\lvert #1 \biggr\rvert}
\newcommand{\absBB}[1]{\Biggl\lvert #1 \Biggr\rvert}

\newcommand{\norm}[1]{\lVert #1 \rVert}

\DeclareMathOperator{\tr}{Tr}

\DeclareMathOperator{\im}{Im}

\newcommand{\ttau}{\vartheta}

\newcommand{\beqa}{\begin{eqnarray}}
\newcommand{\eeqa}{\end{eqnarray}}
\newcommand{\e}{\varepsilon}
\newcommand{\eps}{\varepsilon}

\newcommand{\rd}{{\rm d}}

\newcommand{\bv}{{\bf{v}}}
\newcommand{\bw}{{\bf{w}}}

\newcommand{\be}{\begin{equation}}
\newcommand{\ee}{\end{equation}}

\newcommand{\la}{\lambda}

\newcommand{\cN}{{\mathcal N}}

\newcommand{\bS}{ {\bf  S}}

\newtheorem{theorem}{Theorem}

\newtheorem{lemma}[theorem]{Lemma}
\newtheorem{proposition}[theorem]{Proposition}

\theoremstyle{remark}
\newtheorem{remark}[theorem]{Remark}
\newtheorem{definition}[theorem]{Definition}

\numberwithin{equation}{section}
\numberwithin{theorem}{section}

\title{Eigenvector Distribution of Wigner Matrices}

\author{
 Antti Knowles${}^1$\thanks{Partially supported by NSF grant DMS-0757425} \; and \; Jun Yin${}^2$\thanks{Partially 
supported by NSF grant DMS-1001655} \\\\\\
Department of Mathematics, Harvard University\\
Cambridge MA 02138, USA \\\\ knowles@math.harvard.edu${}^1$ \\ jyin@math.harvard.edu${}^2$ \\ \\}

\begin{document}
\date{November 8, 2011}

\maketitle

\begin{abstract}
We consider $N\times N$ Hermitian or symmetric random matrices
with independent entries. The distribution of the $(i,j)$-th matrix element is given by
a probability measure $\nu_{ij}$ whose first two moments coincide with those of the corresponding Gaussian ensemble.
We prove that the joint probability distribution of the components of eigenvectors associated with eigenvalues close to 
the spectral edge agrees with that of the corresponding Gaussian ensemble. For eigenvectors associated with bulk 
eigenvalues, the same conclusion holds provided the first four moments of the distribution $\nu_{ij}$ coincide with 
those of the corresponding Gaussian ensemble. More generally, we prove that the joint eigenvector-eigenvalue 
distributions near the spectral edge of two generalized Wigner ensembles agree, provided that the first two moments of 
the entries match and that one of the ensembles satisfies a level repulsion estimate. If in addition the first four 
moments match then this result holds also in the bulk.
\end{abstract}

\vspace{1.5cm}


\vspace{0.7cm}


\newpage

\section{Introduction}

The universality of random matrices  can be roughly divided into the bulk universality in the interior of the spectrum 
and the edge universality near the spectral edge.   Over the past two decades,  spectacular  progress on  bulk and  edge 
universality has been made for invariant ensembles, see e.g.\ \cite{BI, DKMVZ1, DKMVZ2, PS} and \cite{AGZ, De1, De2} for 
a review.   For non-invariant ensembles
with i.i.d.\ matrix elements ({\it Standard Wigner ensembles}), edge universality can be  proved via the moment method 
and its various  generalizations; see e.g.\  \cite{SS, Sosh, So1}. In order to establish bulk universality, a new 
approach was developed in a series of papers \cite{ESY1, ESY2, ESY3,  ESY4,  ESYY, EYY, EYY2, EYYrigi} based on three 
basic ingredients:  (1) A local semicircle law -- a precise estimate of the local eigenvalue density  down to energy 
scales containing   around $ N^\e$ eigenvalues.  (2) The eigenvalue distribution of Gaussian divisible ensembles via  an 
estimate on the rate of decay to local equilibrium of the Dyson Brownian motion \cite{Dy}. (3) A density argument which 
shows that for any probability distribution there exists a Gaussian divisible distribution with identical eigenvalue 
statistics down to scales $1/N$. In \cite{EYYrigi}, edge universality is established as a corollary of this approach. It 
asserts that, near the spectral edge, the eigenvalue distributions of two generalized Wigner ensembles are the same 
provided the first two moments of the two ensembles match.

Another approach to both bulk and edge universality was developed in \cite{TV, TV2, J1}. Using this approach, the 
authors show that the eigenvalue distributions of two standard Wigner ensembles are the same in the bulk, provided that 
the first four moments match. They also prove a similar result at the edge, assuming that the first two moments match 
and the third moments vanish.

In this paper, partly based on the approach of \cite{EYYrigi}, we extend edge universality to eigenvectors associated 
with eigenvalues near the spectral edge, assuming the matching of the first two moments of the matrix entries. We prove 
that, under the same two-moment condition as in \cite{EYYrigi}, the edge eigenvectors of Hermitian and symmetric Wigner 
matrices have the same joint distribution as those of the corresponding Gaussian ensembles. The joint distribution of 
the eigenvectors of Gaussian ensembles is well known and can be easily computed.  More generally, we prove that near the 
spectral edge the joint eigenvector-eigenvalue distributions of two generalized Wigner matrix ensembles coincide 
provided that the first two moments of the ensembles match and one of the ensembles satisfies a level repulsion 
condition.

We also prove similar results in the bulk, under the stronger assumption that the first four moments of the two 
ensembles match. In particular, we extend the result of \cite{TV} to cover the universality of bulk eigenvectors.

\subsection{Setup}
We now introduce the basic setup and notations. Let $H^{\f \nu} \equiv H =(h_{ij})_{i,j=1}^N$  be an $N\times N$  
Hermitian or symmetric matrix whose upper-triangular matrix elements $h_{ij}=\bar{h}_{ji}$, $ i \le j$, are independent 
random variables with law $\nu_{ij}$ having mean zero and variance $\sigma_{ij}^2$:
\be
  \E \, h_{ij} =0, \qquad \sigma_{ij}^2:= \E |h_{ij}|^2.
\label{aver}
\ee
The law $\nu_{ij}$ and its variance $\sigma_{ij}^2$ may depend on $N$, but we omit this fact in the notation. We denote by $B \deq (\sigma^2_{ij})_{i,j=1}^N$ the matrix of the variances. We shall always make the following three assumptions on $H$.

\begin{description}
\item[(A)] For any fixed $j$ we have
\be
   \sum_{i=1}^N \sigma^2_{ij} = 1 \, .
\label{sum}
\ee
Thus $B$ is symmetric and doubly stochastic and, in particular, satisfies
$-1\leq B\leq 1$.

\item[(B)]
There exists constants $\delta_- > 0$ and $\delta_+ > 0$, independent of $N$, such that $1$ is a simple eigenvalue of $B$ and
\begin{equation*}
\rr{spec}(B) \;\subset\; [-1 + \delta_-, 1 - \delta_+] \cup \{1\}\,.
\end{equation*}

\item[(C)]
There exists a constant $C_0$, independent of $N$, such that $\sigma_{ij}^2 \leq C_0 N^{-1}$ for all $i,j = 1, \dots, N$.
\end{description}

Examples of matrices satisfying Assumptions \textbf{(A)} -- \textbf{(C)} include Wigner matrices, Wigner matrices whose diagonal elements are set to zero, generalized Wigner matrices, and band matrices whose band width is of order $c N$ for some $c > 0$. See \cite{EYYrigi}, Section 2, for more details on these examples.

In our normalization, the matrix entries $h_{ij}$ have a typical variance of order $N^{-1}$. It is well known that in this 
normalization the empirical eigenvalue density converges to the Wigner semicircle law $\varrho_{sc}(E) \, \dd E$ with 
density
\be
	\varrho_{sc}(E) \;\deq\; \frac{1}{2\pi}
  \sqrt{ (4-E^2)_+} \for E \in \R\,.
\label{def:sc}
\ee
In particular, the spectral edge is located at $\pm 2$. We denote the ordered eigenvalues of $H$ by $\lambda_1 \le 
\ldots \le \lambda_N$, and their associated eigenvectors by $\f u_1, \dots, \f u_N$. The eigenvectors are 
$\ell^2$-normalized. We use the notation $\f u_\alpha = (u_\alpha(i))_{i = 1}^N$ for the components of the vector $\f 
u_\alpha$.

Our analysis relies on a notion of high probability which involves logarithmic factors of $N$. The following definitions 
introduce convenient shorthands.

\begin{definition} \label{def:log factors}
We set $L \equiv L_N \deq A_0 \log \log N$ for some fixed $A_0$ as well as $\varphi \equiv \varphi_N \deq (\log N)^{\log 
\log N}$.
\end{definition}

\begin{definition}
We say that an $N$-dependent event $\Omega$ holds \emph{with high probability} if $\P(\Omega) \geq 1 - \me^{- 
\varphi^c}$ for large enough $N$ and some $c > 0$ independent of $N$.
\end{definition}

A key assumption for our result is the following level repulsion condition, which is in particular satisfied by the 
Gaussian ensembles (see Remark \ref{rem: level repulsion} below). Consider a spectral window whose size is much smaller 
than the typical eigenvalue separation.  Roughly, the level repulsion condition says that the probability of finding 
more than one eigenvalue in this window is much smaller than the probability of finding precisely one eigenvalue.  In 
order to state the level repulsion condition, we introduce the following counting function.  For any $E_1\le E_2$ we 
denote the number of eigenvalues in $[E_1, E_2]$ by
$$
	\cN(E_1, E_2) \;\deq\; \#\{j \st E_1\le \la_j \leq E_2\}\,.
$$

\begin{definition}[Level repulsion at the edge] \label{def: level repulsion}
The ensemble $H$ is said to satisfy \emph{level repulsion at the edge} if, for any $C > 0$, there is an $\alpha_0 > 0$ 
such that the following holds. For any $\alpha$ satisfying $0 < \alpha \le \alpha_0$ there exists a $\delta > 0$ such 
that
\be\label{lre}
\P \left (   \cN (E - N^{-2/3 - \alpha} , E + N^{-2/3 - \alpha})  \ge 2 \right  ) \;\leq\; N^{-\alpha-\delta}
\ee
for all $E$ satisfying $ |E + 2| \le N^{-2/3} \varphi^C$.
\end{definition}

\begin{definition}[Level repulsion in the bulk] \label{def: level repulsion bulk}
The ensemble $H$ is said to satisfy \emph{level repulsion in the bulk} if, for any $\kappa > 0$, there is an $\alpha_0 > 
0$ such that the following holds. For any $\alpha$ satisfying $0 < \alpha \le \alpha_0$ there exists a $\delta > 0$ such 
that
\be\label{lrb}
\P \left (   \cN (E - N^{-1 - \alpha} , E + N^{-1 - \alpha})  \ge 2 \right  ) \;\leq\; N^{-\alpha-\delta}
\ee
for all $E \in [-2 + \kappa, 2 - \kappa]$.
\end{definition}

\begin{remark} \label{rem: level repulsion}
Both the Gaussian Unitary Ensemble (GUE) and the Gaussian Orthogonal Ensemble (GOE) satisfy level repulsion in sense of 
Definitions \ref{def: level repulsion} and \ref{def: level repulsion bulk}. This can be established for instance as 
follows; see \cite{AGZ}, Sections 3.5 and 3.7, and in particular Lemmas 3.5.1 and 3.7.2, for full details. For GUE and 
GOE, the correlation functions can be explicitly expressed in terms of Hermite polynomials. Using Laplace's method, one 
may then derive the large-$N$ asymptotics of the correlation functions, from which \eqref{lre} and \eqref{lrb} 
immediately follow. (Note that in \cite{AGZ}, the exponent of $N$ in the error estimates  was not tracked in order to 
simplify the presentation.)

In the more general case of Wigner matrices, level repulsion in the bulk, \eqref{lrb}, was proved for  matrices with 
smooth distributions in \cite{ESY3} and without a smoothness assumption in \cite{TV}.
\end{remark}

We shall use the level repulsion condition of Definition \ref{def: level repulsion} to estimate the probability of 
finding two eigenvalues closer to each other than the typical eigenvalue separation. For definiteness, we formulate this 
estimate at the lower spectral edge $-2$. By partitioning the interval
\begin{equation*}
\qB{-2- N^{-2/3} \varphi^C\,,\, -2 + N^{-2/3} \varphi^C}
\end{equation*}
into $O(\varphi^C N^\alpha)$ subintervals of size $N^{-2/3 - \alpha}$, we get from \eqref{lre} that for any sufficiently 
small $\alpha$ there exists a $\delta > 0$ such that
\be\label{lr1}
\P \left ( \text{there exists $E$  with } |E + 2| \le N^{-2/3} \varphi^C  \text{ such  that }   \cN(E - N^{-2/3 - 
\alpha} , E + N^{-2/3 - \alpha})  \ge 2 \right  ) \;\leq\; N^{-\delta}\,.
\ee
A similar result can be derived in the bulk using \eqref{lrb}.

\subsection{Results}
Before stating our main results, we recall the definition of the classical eigenvalue locations.
Let
\be
n_{sc}(E) \;\deq\; \int_{-\infty}^E \varrho_{sc}(x) \, \dd x
\label{nsc}
\ee
be the integrated distribution function of the semicircle law. We use $\gamma_\alpha \equiv \gamma_{\alpha,N}$ to denote 
the classical location of the $\alpha$-th eigenvalue under the semicircle law, defined through
\begin{align} \label{def:gamma}
n_{sc}(\gamma_\alpha) \;=\; \frac{\alpha}{N}\,.
\end{align}

To avoid unnecessary technicalities in the presentation, we shall assume that the entries $h_{ij}$ of $H$ have uniform 
subexponential decay, i.e.
\be\label{subexp}
\P(|h_{ij}|\geq x \sigma_{ij}) \;\leq\; \ttau^{-1} \exp ( - x^\ttau)
\ee
$\vartheta > 0$ is some fixed constant. As observed in \cite{EKYY2}, Section 7, one may easily check that all of our
results hold provided the subexponential condition \eqref{subexp} is replaced with the weaker assumption that there is a 
constant $C$ such that
\begin{equation*}
\E \absb{h_{ij} \sigma_{ij} ^{-1}}^{C_0} \;\leq\; C\,,
\end{equation*}
where $C_0$ is a large universal constant.

Our main result on the distributions of edge eigenvectors is the following theorem.

\begin{theorem}[Universality of edge eigenvectors] \label{t1}
Let $H^{\f v}$ and $H^{\f w}$ both satisfy Assumptions {\bf(A)} -- {\bf (C)} as well as the uniform subexponential decay condition \eqref{subexp}.
Let $\E^{\f v}$ and $\E^{\f w}$ denote the expectations with respect to these collections of random variables.  Suppose 
that the level repulsion estimate \eqref{lre} holds for the ensemble $H^{\f v}$. Assume that
the first two moments of the entries of $H^\bv$ and $H^\bw$ are the same, i.e.
\be\label{2m}
	 \E^{\f v} \bar h_{ij}^l h_{ij}^{u} \;=\; \E^{\f w} \bar h_{ij}^l h_{ij}^{u}
  \for 0\le l+u\le 2\,.
\ee
Let $\rho$ be a positive constant. Then for any integer $k$ and any choice of indices $i_1, \ldots i _k$,  $j_1, \ldots, 
j_k$, $\beta_1, \ldots \beta_k$  and $\alpha_1, \ldots \alpha_k$ with $ \min (|\alpha_l |, |\alpha_l- N| ) +  \min 
(|\beta_l |, |\beta_l- N| )  \le \varphi_N^\rho$ for all $l$ we have \be\label{11}
 \lim_{N\to \infty} \big [ \E^{\f v} - \E^{{\f w}}   \big ]  \theta \pB{N^{2/3}(\lambda_{\beta_1} - \gamma_{\beta_1}), 
\ldots, N^{2/3}(\lambda_{\beta_k} - \gamma_{\beta_k}) \,;\,  N  \bar u_{\alpha_1} (i_1)  u_{\alpha_1} (j_1),    \ldots, 
N  \bar u_{\alpha_k} (i_k)  u_{\alpha_k} (j_k)} \;=\; 0\,,
\ee
where $\theta$ is a smooth function that satisfies
\begin{equation}
\abs{\partial^n \theta(x)} \;\leq\; C (1 + \abs{x})^C
\end{equation}
for some arbitrary $C$ and all $n \in \N^{2k}$ satisfying
$\abs{n} \leq 3$.
The convergence is uniform in all the parameters $i_l,j_l,\alpha_l,\beta_l$ satisfying the above conditions.  
\end{theorem}

\begin{remark}
The scaling in front of the arguments in \eqref{11} is the natural scaling near the spectral edge. Indeed, for e.g.\ GUE 
or GOE it is known (see e.g.\ \cite{AGZ}) that $(\lambda_\beta - \gamma_\beta) \sim N^{-2/3}$ near the edge, and that 
$u_\alpha(i) \sim N^{1/2}$ (complete delocalization of eigenvectors).
\end{remark}

\begin{remark}
The form \eqref{11} characterizes the distribution of the edge eigenvectors completely. Choosing $i_l = j_l$ yields the 
modulus $\abs{u_{\alpha_l}(i_l)}^2$; fixing $i_l$ and varying $j_l$ gives the relative phases of the entries of the 
vector $\f u_{\alpha_l}$, which is only defined up to a global phase.
\end{remark}

\begin{remark}
Theorem \ref{t1} and Remark \ref{rem: level repulsion} imply that the joint eigenvector-eigenvalue distribution of 
Hermitian Wigner matrices agrees with that of GUE. In the case of GUE, it is well known that the joint distribution of 
the eigenvalues is given by the Airy kernel \cite{TW}. The eigenvectors are independent of the eigenvalues, and the 
matrix $(u_\alpha(i))_{\alpha, i}$ of the eigenvector entries is distributed according to the Haar measure on the 
unitary group $U(N)$. In particular, any eigenvector $\f u_\alpha$ is uniformly distributed on the unit $(N - 
1)$-sphere.

Similarly, Theorem \ref{t1} and Remark \ref{rem: level repulsion} imply that the joint eigenvector-eigenvalue 
distribution of symmetric Wigner matrices agrees with that of GOE. Results similar to those outlined above on the 
eigenvector-eigenvalue distribution of GUE hold for GOE.
\end{remark}

The universality of the eigenvalue distributions near the edge was already proved in \cite{EYYrigi} under the assumption 
that the first two moments of the matrix entries match, and in \cite{TV2} under the additional assumption that the third 
moments vanish.  Note that Theorem \ref{t1} holds in a stronger sense than the result in \cite{EYYrigi}: it holds for 
probability density functions, not just the distribution functions.

In the bulk, a result similar to Theorem \ref{t1} holds under the stronger assumption that four, instead of two, moments 
of the matrix entries match.

\begin{theorem}[Universality of bulk eigenvectors] \label{t2}
Let $H^{\f v}$ and $H^{\f w}$ both satisfy Assumptions {\bf (A)} -- {\bf (C)} as well as the uniform subexponential decay condition \eqref{subexp}.  Suppose that the level repulsion estimate \eqref{lrb} holds 
for the ensemble $H^{\f v}$. Suppose moreover that the first four off-diagonal moments of
 $H^\bv$ and $H^\bw$ are the same, i.e.
\be\label{2m1}
	 \E^{\f v} \bar h_{ij}^l h_{ij}^{u} =  \E^{\f w} \bar h_{ij}^l h_{ij}^{u}
\for i \neq j \quad \text{and} \quad 0\le l+u\le 4\,,
\ee
and that the first two diagonal moments of
 $H^\bv$ and $H^\bw$ are the same, i.e.
\be\label{2m2}
	 \E^{\f v} \bar h_{ii}^l h_{ij}^{u} =  \E^{\f w} \bar h_{ii}^l h_{ij}^{u}
\for 0\le l+u\le 2\,.
\ee

Let $\rho > 0$ be fixed.
Then for any integer $k$ and any choice of indices $i_1, \ldots i _k$, $j_1, \ldots, j_k$, as well as $\rho N \leq
\alpha_1, \ldots \alpha_k, \beta_1, \ldots, \beta_k \leq (1 - \rho)N$, we have
\be\label{12}
 \lim_{N\to \infty} \big [ \E^{\f v} - \E^{{\f w}}   \big ]  \theta \pB{ N(\lambda_{\beta_1} - \gamma_{\beta_1}), \dots, 
N(\lambda_{\beta_k} - \gamma_{\beta_k}) \,;\, N  \bar u_{\alpha_1} (i_1)  u_{\alpha_1} (j_1),    \dots, N  \bar 
{u}_{\alpha_k} (i_k)  {u}_{\alpha_k}  (j_k)} \;=\; 0\,,
\ee
where $\theta$ is a smooth function that satisfies
\begin{equation} \label{growth of derivatives}
\abs{\partial^n \theta(x)} \;\leq\; C (1 + \abs{x})^C
\end{equation}
for some arbitrary $C$ and all $n \in \N^{2k}$ satisfying
$\abs{n} \leq 5$.
The convergence is uniform in all the parameters $i_l,j_l,\alpha_l,\beta_l$ satisfying the above conditions.
\end{theorem}

The universality restricted to the bulk eigenvalues only has been previously established in several works. The following 
list provides a summary. Note that the small-scale statistics of the eigenvalues may be studied using \emph{correlation 
functions}, which depend only on eigenvalue differences, or using \emph{joint distribution functions}, as in \eqref{11} 
and \eqref{12}, which in addition contain information about the eigenvalue locations.
\begin{enumerate}
\item
In \cite{EYYrigi}, bulk universality for generalized Wigner matrices was proved in the sense that correlation functions 
of bulk eigenvalues, averaged over a spectral window of size $N^\epsilon$, converge to those of the corresponding 
Gaussian ensemble.
\item
In \cite{TV} the statement \eqref{12} on distribution functions, restricted to eigenvalues only, was proved for 
Hermitian and symmetric Wigner matrices for the case where the first four moments match as in \eqref{2m1}.
\item
For the case of Hermitian Wigner matrices with a finite Gaussian component, it was proved in \cite{J} that the
correlation functions converge to those of GUE.
\item
In \cite{GEG1}, the joint distribution function of the eigenvalues of GUE was computed. This result was extended to 
cover GOE in \cite{GEG2}.
\end{enumerate}

Note that (ii) and (iii) together imply the universality of the joint distribution of eigenvalues for Hermitian Wigner 
matrices, for which the first three moments match those of GUE and the distribution is supported on at least three 
points. Moreover, combining (ii) and (iv) allows one to compute the eigenvalue distribution of Hermitian and symmetric 
Wigner matrices, provided the four first moments match those of GUE/GOE.

Thus, Theorem \ref{12} extends the results of \cite{EYYrigi} to distribution functions of individual eigenvalues as well 
as to eigenvectors.

\begin{remark}
A while after this paper was posted online, a result similar to Theorem \ref{t2} appeared in \cite{TV3}. Its proof 
relies on a different method. The hypotheses of \cite{TV3} are similar to those of Theorem \ref{t2}, with
the two following exceptions. The result of \cite{TV3} is restricted to Wigner matrices instead of the generalized 
Wigner matrices defined by Assumptions {\bf (A)} -- {\bf (C)}.  Moreover, in \cite{TV3} the derivatives of the observable $\theta$ are required to be uniformly bounded
in $x$, where this uniform bound may grow slowly with $N$. This latter restriction allows the authors of \cite{TV3} to let $k$ grow slowly with $N$.

While the results of \cite{TV3} apply to eigenvectors near the spectral edge, the matching of four moments (as in 
Theorem \ref{t2}) is also required for this case. As shown in Theorem \ref{t1}, the universality of edge eigenvectors in 
fact only requires the first two moments to match.
\end{remark}

\subsection{Outline of the proof} \label{sect: outline of proof}
The main idea behind our proof is to express the eigenvector components using matrix elements of the Green function 
$G(z) = (H - z)^{-1}$. To this end, we use the identity
\begin{equation} \label{main identity}
\sum_{\beta} \frac{\eta / \pi}{(E - \lambda_\beta)^2 + \eta^2} \, N \bar u_\beta(i) u_\beta(j) \;=\; \frac{N}{2 \pi \ii} 
\pb{G_{ij}(E + \ii \eta) - G_{ij}(E - \ii \eta)}\,,
\end{equation}
where $\eta > 0$. Using a good control on the matrix elements of $G(z)$, we may then apply a Green function comparison 
argument (similar to the Lindeberg replacement strategy) to complete the proof. For definiteness, let us consider a 
single eigenvalue $\lambda_\alpha$ located close to the spectral edge $-2$.

In a first step, we write $N \bar u_\alpha(i) u_\alpha(j)$ as an integral of \eqref{main identity} over an appropriately 
chosen (random) domain, up to a negligible error term. We choose $\eta$ in \eqref{main identity} to be much smaller than 
the typical eigenvalue separation, i.e.\ we set $\eta = N^{-2/3 - \epsilon}$ for some small $\epsilon > 0$. Note that 
the fraction on the left-hand side of \eqref{main identity} is an approximate delta function on the scale $\eta$. Then 
the idea is to integrate \eqref{main identity} over the interval $[\lambda_\alpha - \varphi^C \eta, \lambda_\alpha + 
\varphi^C \eta]$ for some large enough constant $C$. For technical reasons related to the Green function comparison (the 
third step below), it turns out to be advantageous to replace the above interval with $\cal I \deq [\lambda_{\alpha - 
1} + \varphi^C \eta, \lambda_\alpha + \varphi^C \eta]$. Using eigenvalue repulsion, we infer that, with sufficiently 
high probability, the eigenvalues $\lambda_{\alpha - 1}$ and $\lambda_{\alpha + 1}$ are located at a distance greater 
than $\varphi^C \eta$ from $\lambda_\alpha$. Therefore the $E$-integration over $\cal I$ of the right-hand side of 
\eqref{main identity} yields $N \bar u_{\alpha}(i) u_\alpha(j)$ up to a negligible error term.

In a second step, we replace the sharp indicator function $\ind{E \in \cal I}$ with a smoothed indicator function 
expressed in terms of the Green function $G$.  This is necessary for the Green function comparison argument, which 
requires all $H$-dependence to be expressed using Green functions. To that end, we choose a scale $\tilde \eta \deq 
N^{-2/3 - 6 \epsilon} \ll \eta$ and write
\begin{equation} \label{approx indicator function}
\ind{E \in \cal I} \;\approx\; q \qB{ \tr  \pb{{\bf 1}_{[E_L, E - \varphi^C \eta]}  \ast \theta_{\tilde \eta}} (H)}
\end{equation}
where the error is negligible. Here $E_L = - 2 - \varphi^C N^{-2/3}$, $q$ is a smooth function equal to $1$ in the 
$1/3$-neighbourhood of $\alpha - 1$ and vanishing outside the $2/3$-neighbourhood of $\alpha - 1$, and $\theta_{\tilde 
\eta}$ is the approximate delta function defined in \eqref{thetam} below. Thanks to the special form of the right-hand 
side of \eqref{thetam}, we have $\theta_\eta(H) = \frac{1}{\pi} \im G(\ii \eta)$. Hence the argument on right-hand side 
of \eqref{approx indicator function} may be expressed as an integral over Green functions. Thus we have expressed $N 
\bar u_\alpha(i) u_\alpha(j)$ using matrix elements of $G$ alone.  Note that the above choice of $\cal I$ was made 
precisely so as to make the right-hand side of \eqref{approx indicator function} a simple function of $G$.

In a third step, we use a Green function comparison argument to compare the distributions of $N \bar u_\alpha(i) 
u_\alpha(j)$ under the two ensembles $H^{\f v}$ and $H^{\f w}$. The basic strategy is similar to
\cite{EYYrigi}, but requires a more involved analysis of the resolvent expansion. The reason for this is that we need to 
exploit the smallness associated with off-diagonal elements of $G$, which requires us to keep track of their number in 
the power counting.  This bookkeeping is complicated by the presence of the two fixed indices $i$ and $j$. Another 
important ingredient in the error estimates of the Green function comparison argument is the restriction of the 
integration over $E$ to a deterministic interval of size $\varphi^C N^{-2/3}$ around $-2$. This can be done with 
negligible errors using the eigenvalue rigidity proved in \cite{EYYrigi}; see Theorem \ref{7.1}.

The above proof may be easily generalized to multiple eigenvector components as well as to eigenvalues; this allows us 
to consider observables of the form given in \eqref{11}.  The necessary changes are given in Section \ref{sect:general 
edge}.

The proof for bulk eigenvectors is similar, with two major differences. At the edge, the convolution integral on the 
right-hand side of \eqref{approx indicator function} was over a domain of size $\varphi^C N^{-2/3}$. If the same 
expression were used in the bulk, this size would be $O(1)$ (since $E$ is separated from the spectral edge $-2$ by a 
distance of order $O(1)$), which is not affordable in the error estimates. Instead, a more refined multiscale approach 
using the Helffer-Sj\"ostrand functional calculus is required in order to rewrite the sharp indicator function on the 
left-hand side of \eqref{approx indicator function} in terms of Green functions.  The second major difference for bulk 
eigenvectors is the power counting in the Green function comparison argument, which is in fact easier than at the edge.  
The main reason for this is that the smallness associated with off-diagonal elements of $G$ is not available in the 
bulk. Hence we need to assume that four instead of two moments match, and the intricate bookkeeping of the number of 
off-diagonal resolvent elements is not required. Thus, thanks to the very strong assumption of four-moment matching, the 
proof of Theorem \ref{t2} is considerably simpler than that of Theorem \ref{t1}.
See Section \ref{section: bulk} for a more detailed explanation as well the proof.

\medskip \noindent
{\bf Conventions.} We shall use the letters $C$ and $c$ to denote generic positive constants, which may depend on fixed 
quantities such as $\vartheta$ from \eqref{subexp}, $\delta_\pm$ from Assumption {\bf (B)}, and $C_0$ from Assumption {\bf (C)}. We use $C$ for large constants 
and $c$ for small constants.

\medskip
\noindent
{\bf Acknowledgements.} The authors would like to thank L.\ Erd\H{o}s and H.T.\ Yau for many insights and helpful 
discussions.

\section {Local semicircle law and rigidity of eigenvalues} \label{sec: loc sc}

In this preliminary section we collect the main tools we shall need for our proof. We begin by introducing some notation 
and by recalling the basic results from \cite{EYYrigi} on the local semicircle law
and the rigidity of eigenvalues.

We define the Green function of $H$ by
\be\label{green}
G_{ij}(z) \;=\; \left(\frac1{H-z}\right)_{ij}\,,
\ee
where we the spectral parameter $x = E + \ii \eta$ satisfies $E \in \R$ and $\eta > 0$.
The Stieltjes transform of the empirical eigenvalue distribution of $H$ is defined as
\be
m(z) \;\deq\;   \frac{1}{N}    \sum_i G_{ii}(z) \;=\; \frac{1}{N} \tr\, \frac{1}{H-z} \;=\; \frac{1}{N} \sum_\alpha 
\frac{1}{\lambda_\alpha - z}\,.
\label{mNdef}
\ee
Similarly, we define $m_{sc}(z)$ as the Stieltjes transform of the local semicircle law:
\begin{equation*}
m_{sc}(z) \;\deq\; \int \frac{\varrho_{sc}(\lambda) \, \dd \lambda}{\lambda - z}\,.
\end{equation*}
It is well known that $m_{sc} (z)$ can also be characterized as the unique solution of
\be\label{defmsc} m_{sc} (z)  + \frac{1}{z+m_{sc} (z)} \;=\; 0
\ee
with positive imaginary part for all $z$ with $\im z > 0$. Thus,
\be\label{temp2.8}
m_{sc}(z) \;=\; \frac{-z+\sqrt{z^2-4}}{2}\,,
\ee
where the square root function is chosen with a branch cut in the segment
$[-2,2]$ so that asymptotically $\sqrt{z^2-4}\sim z$ at infinity.
This guarantees that the imaginary part of $m_{sc}$ is non-negative for   $\eta=\im  z > 0$ and in the limit $\eta\to 
0$.

In order to state the local semicircle law, we introduce the control parameters
\be\label{defLambda}
  \Lambda_d \;\deq\; \max_i |G_{ii}-m_{sc}|\,, \qquad
 \Lambda_o \;\deq\; \max_{i\ne j} |G_{ij}|\,, \qquad \Lambda \;\deq\; |m-m_{sc}|\,,
\ee
where the subscripts refer to ``diagonal'' and ``off-diagonal'' matrix elements.  All these quantities depend on the 
spectral parameter $z$ and on $N$, but for simplicity we often omit the explicit mention of this dependence from the 
notation. The following two results were proved in \cite{EYYrigi}.

\begin{theorem}[Strong local semicircle law] \label{45-1} Let $H=(h_{ij})$ be a Hermitian or symmetric 
$N\times N$ random matrix satisfying Assumptions {\bf A} -- {\bf C}. Suppose that the distributions of the matrix 
elements $h_{ij}$ have a uniformly subexponential decay in the sense of \eqref{subexp}. Then  there exist positive constants $A_0 
> 1$, $C, c$, and $\tau < 1$,  such that the following estimates hold for $L$ as in Definition \ref{def:log factors} and 
for $N\ge N_0(\ttau, C_0, \delta_\pm)$ large enough.

\begin{enumerate}
\item
The Stieltjes transform of the empirical eigenvalue distribution of  $H $  satisfies
\be\label{Lambdafinal}
\P \pBB{\bigcup_{z\in \bS_L} \hbb{ \Lambda(z) \ge \frac{(\log N)^{4L}}{N\eta}}} \;\le\; \me^{-c (\log N)^{\tau L}}\,,
\ee
where
\be
{\bf  S}_L \;\deq\; \Big\{ z=E+i\eta\; : \;
 |E|\leq 5,  \quad  N^{-1}(\log N)^{10L} < \eta \le  10  \Big\}\,.
\label{defS}
\ee
\item
The individual  matrix elements of
the Green function  satisfy
\be\label{Lambdaodfinal}
\P \left  ( \bigcup_{z\in \bS_L} \hBB{\Lambda_d(z)  + \Lambda_o (z) \geq (\log N)^{4L} \sqrt{\frac{\im m_{sc}(z)  
}{N\eta}} + \frac{(\log N)^{4L}}{N\eta}}    \right)
\;\leq\; \me^{-c (\log N)^{\tau L}}\,.
\ee
\item
The norm of $H$ is bounded by $2+N^{-2/3}(\log N)^{ 9L} $ in the sense that
\be\label{443}
\P \Big (  \norm{H} \ge 2+N^{-2/3}(\log N)^{ 9 L} \Big) \;\le\;
\me^{-c (\log N)^{\tau L}}\,.
\ee
\end{enumerate}
\end{theorem}

The local semicircle law implies that the eigenvalues are close to their classical locations with high probability.  
Recall that $\lambda_1 \leq \lambda_2 \leq \cdots \leq \lambda_N$ are the ordered eigenvalues of $H$. The classical 
location $\gamma_\alpha$ of the $\alpha$-th eigenvalue was defined in \eqref{nsc}.

\begin{theorem} [Rigidity of eigenvalues] \label{7.1}
Under the assumptions of Theorem \ref{45-1} there exist positive constants $A_0 > 1$,  $C, c$, and $\tau < 1$,  
depending only on $\ttau$ in \eqref{subexp}, $\delta_\pm$ in Assumption {\bf (B)}, and $C_0$ in Assumption {\bf (C)}, such that
such that
\be\label{rigidity}
\P \hB{  \exists \, \alpha \st \abs{\lambda_\alpha-\gamma_\alpha} \ge (\log N)^{ L}  \qb{ \min (\alpha ,  N-\alpha+1 
)}^{-1/3}   N^{-2/3}}
 \;\le\; \me^{-c (\log N)^{\tau L}}\,,
\ee
where $L$ is given in Definition \ref{def:log factors}.
\end{theorem}

A simple consequence of Theorem \ref{45-1} is that the eigenvectors of $H$ are completely delocalized.

\begin{theorem}[Complete delocalization of eigenvectors] \label{DTH}
Under the assumptions of Theorem \ref{45-1} we have
\be\label{71.1}
\P \hbb{\exists\, \al, i \st  |u_\al(i)|^2\geq \frac{\varphi^C}{N} }
 \;\le\; \me^{-c (\log N)^{\tau L}}
\ee
for some positive constants $C$ and  $c$.
\end{theorem}

\begin{proof}
Using \eqref{Lambdaodfinal} and \eqref{443} we have, with probability greater than $1 - \me^{-c (\log N)^{\tau L}}$,
\begin{align*}
C \;\geq\; \im G_{ii}(\lambda_\alpha + \ii \eta) \;=\; \sum_{\beta} \frac{\eta \abs{u_\beta(i)}^2}{(\lambda_\alpha - 
\lambda_\beta)^2 + \eta^2} \;\geq\; \frac{\abs{u_\alpha(i)}^2}{\eta}\,.
\end{align*}
Choosing $\eta = N^{-1} (\log N)^{20 L}$ yields the claim.
\end{proof}

\subsection{Stability of the level repulsion condition}
In this section we prove that level repulsion, in the sense of \eqref{lre} (respectively \eqref{lrb}), holds for the 
ensemble $H^{\f w}$ provided it holds for the ensemble $H^{\f v}$ and the first two (respectively four) moments of the 
entries of $H^{\f v}$ and $H^{\f w}$ match.

\begin{proposition}[Stability of level repulsion at the edge] \label{leml}
Let $H^{\f v}$ and $H^{\f w}$ both satisfy Assumptions {\bf(A)} -- {\bf (C)} as well as the uniform subexponential decay condition \eqref{subexp}. Assume moreover that the first two moments of the entries of $H^\bv$ and $H^\bw$ are the same, in the sense of \eqref{2m}.
If the level repulsion estimate \eqref{lre} holds for $H^{\f v}$ then it holds for $H^{\f w}$.
\end{proposition}

\begin{proposition}[Stability of level repulsion in the bulk] \label{lemlb}
Let $H^{\f v}$ and $H^{\f w}$ both satisfy Assumptions {\bf(A)} -- {\bf (C)} as well as the uniform subexponential decay condition \eqref{subexp}. Assume moreover that the first four moments of the entries of $H^\bv$ and $H^\bw$ are the same, in the sense of \eqref{2m1} and \eqref{2m2}.
If the level repulsion estimate \eqref{lrb} holds for $H^{\f v}$ then it holds for $H^{\f w}$.
\end{proposition}

The proofs of Propositions \ref{leml} and \ref{lemlb} are very similar. For definiteness, we give the details for the 
edge case (Proposition \ref{leml}). The rest of this section is devoted to the proof of Proposition \ref{leml}.  The 
main tool is the following Green function comparison theorem, which was proved in \cite{EYYrigi}, Theorem 6.3.

\begin{lemma} [Green function comparison theorem at the edge] \label{GFCT}
Suppose that  the assumptions of Theorem  \ref{45-1} hold for both ensembles $H^{\f v}$ and $H^{\f w}$.
Let $F:\R \to \R$ be a function whose derivatives satisfy
\be\label{gflowder}
\max_{x} \, \absb{F^{(n)}(x)} \, (1 + \abs{x})^{-C_1} \;\leq\; C_1 \for n=1,2,3,4\,,
\ee
with some constant $C_1>0$.
Then there exists a constant $\e_0>0$, depending only on $C_1$, such that for any $\e<\e_0$
and for any real numbers $E_1$ and $E_2$ satisfying
$$
|E_1+2|\leq N^{-2/3+\eps}, \qquad |E_2+2|\leq N^{-2/3+\eps}, \qquad
$$
we have
\be\label{c1}
\absBB{ \qb{\E^{\f v} - \E^{\f w}}  F \pbb{N \int_{E_1}^{E_2} \rd y \;  \im m \pb{y + \ii N^{-2/3 - \epsilon}}}} \;\leq\; 
C N^{-1/6+C \e}
\ee
and
\begin{equation} \label{c2}
\absB{\qb{\E^\bv - \E^\bw} F\pB{N \eta \, \im m(E_1 + \ii N^{-2/3 + \epsilon})}} \;\leq\; C N^{-1/6 + C \epsilon}\,,
\end{equation}
for some constant $C$ and large enough $N$, depending only on $C_1$, $\ttau$ in \eqref{subexp}, $\delta_\pm$ in Assumption {\bf (B)}, and $C_0$ in Assumption {\bf (C)}.
\end{lemma}

The basic idea behind the proof of Proposition \ref{leml} is to first cast the level repulsion estimate into an 
estimate in terms of Green functions and then use the Green function comparison theorem.
Recalling $L$ from Definition \ref{def:log factors},
we set
\be\label{defEL}
E_L \;\deq\; -2-2(\log N)^{L}N^{-2/3}\,.
\ee
For any $E \geq E_L$ let
\[
\chi_E \;\deq\; {\bf 1}_{[E_L, E]}
\]
be the characteristic function of the interval $[E_L, E]$. For
any $\eta>0$ we define the approximate delta function $\theta_\eta$ on the scale $\eta$ through
\be\label{thetam}
\theta_\eta(x) \;\deq\; \frac{\eta }{\pi(x^2+\eta^2)} \;=\; \frac{1}{\pi} \im \frac{1}{x-i\eta}\,.
\ee
The following result provides a tool for estimating the number operator using Green functions. It is proved in 
\cite{EYYrigi}, Lemma 6.1 and Corollary 6.2.

\begin{lemma}\label{lem:21}
Suppose that the assumptions of Theorem \ref{45-1} hold, and let $A_0$ and $\tau$ be as in Theorem \ref{45-1}.
For any  $\e>0$, set $\ell_1:=N^{-2/3-3\e}$ and $\eta:=N^{-2/3-9\e}$. Then there exist constants $C, c$ such that,
for any $E$ satisfying
\be\label{E-2N}
\abs{E+2} \, N^{2/3} \;\leq\; \frac 3 2(\log N)^L\,,
\ee
we have
\be\label{6.10}
\absb{\tr \chi_E(H) - \tr  \chi_E  \ast \theta_\eta (H)}  \;\le\; C\left( N^{-2\e}  +   \cN (E-\ell_1, E+\ell_1)  
\right)
\ee
\hp.

Moreover, let $\ell: = \frac{1}{2}\ell_1 N^{ 2\e} = \frac{1}{2}N^{-2/3 - \e}$. Then under the above assumptions the 
inequalities
\be\label{41new}
\tr (\chi_{E - \ell}  \ast \theta_\eta) (H)  -  N^{-\e} \;\le\;  \cN (-\infty, E)  \;\le\;  \tr  (\chi_{E + \ell  }  
\ast \theta_\eta) (H)  +  N^{-\e}
\ee
hold \hp.
\end{lemma}

After these preparations we may complete the proof of Proposition \ref{leml}.

\begin{proof}[Proof of Proposition \ref{leml}]
Assume that $H^{\f v}$ satisfies the level repulsion assumption \eqref{lre} with constant $\alpha_0$. We shall show that 
$H^{\f w}$ satisfies also satisfies Definition \eqref{lre} with the same constant $\alpha_0$. Fix $\alpha$ satisfying $0 
< \alpha \leq \alpha_0$, and let $\delta > 0$ be as chosen so that \eqref{lre} holds for the ensemble $H^{\f v}$.

Abbreviate $E_\pm = E\pm N^{-2/3 - \alpha}$ and set $\e \deq 2 \alpha$. By using \eqref{41new} for $E =E_+$ and $E = 
E_-$, and subtracting the resulting two inequalities, we get, \hp,
\be  \label{411}
 \tr ({\bf 1}_{[E_- + \ell,  E_+ - \ell] }   \ast \theta_\eta) (H)  -  2 N^{-\e}  \;\le\;  \cN (E _-, E _+)  \;\le\;  
\tr ({\bf 1}_{[E_- - \ell  , E_+ + \ell]  }  \ast \theta_\eta) (H)   +  2 N^{-\e}\,.
\ee
Let $F$ be a nonnegative increasing smooth function satisfying $F (x) = 1$ for $x \ge 2$ and $F(x) = 0$ for $ x \le 
3/2$. Then, using \eqref{411} and Lemma \ref{GFCT}, we have
\begin{align*} 
\E^{\f w}  F(\cN (E _-, E _+) ) & \;\le\;  \E^{\f w} F\pb{  \tr ({\bf 1}_{[E_- - \ell  , E_+ + \ell]  }  \ast \theta_\eta) 
(H) +  2 N^{-\e}} \\
& \;\le\; \E^{\f v}  F \pb{ \tr  ({\bf 1}_{[E_- - \ell  , E_+ + \ell]  }  \ast \theta_\eta) (H)}  +  C N^{-\e} + C 
N^{-1/6+ C \e} \\
& \;\le\; \E^{\f v}  F (\cN (E _- - 2 \ell, E _+ + 2 \ell ) + N^{-\e} )  +  C N^{-\e} + C N^{-1/6+ C \e}
\\ & \;\le\; N^{-\alpha-\delta  } +   C N^{-\e} + C N^{-1/6+ C \e}\,.
\end{align*}
Since $\e = 2 \alpha$, we get that \eqref{lre} holds for the ensemble $H^{\f w}$ with exponent $\delta' = \min\h{\delta, 
\alpha} > 0$.
\end{proof}

\section{Proof of Theorem \ref{t1}} \label{section: proof}

To simplify presentation, in this section we prove Theorem \ref{t1} in the special case $\theta = \theta \pb{ N \bar 
u_\alpha(i) u_\alpha(j)}$, where $\alpha \leq \varphi^\rho$. The proof of the general case is analogous; see Section 
\ref{sect:general edge} for more details.

In a first step we convert the eigenvector problem into a problem involving the Green function $G_{ij}$. To that end, we
define
\begin{align} \label{def of tilde G}
\wt G_{ij}(z) \;\deq\; \frac{1}{2 \ii} \pb{G_{ij}(z) - G_{ij}(\bar z)} \;=\; \eta \sum_k G_{ik}(z) \ol{G_{jk}(z)}
\,,
\end{align}
where the second equality follows easily by spectral decomposition, $G_{ij}(z) = \sum_\beta \frac{\bar u_\beta(i) 
u_\beta(j)}{\lambda_\beta - z}$.
Note that
\begin{align*}
\wt G_{ij}(E + \ii \eta) \;=\; \sum_\beta \frac{\eta}{(E - \lambda_\beta)^2 + \eta^2} \, \bar u_\beta(i) u_\beta(j)
\end{align*}
as well as
$\im G_{ii}(z) = \wt G_{ii}(z)$.
It is a triviality that all of the results from Section \ref{sec: loc sc} hold with $z$ replaced with $\bar z$.

The following lemma expresses the eigenvector components as an integral of the Green function over an appropriate random 
interval.

\begin{lemma}\label{CTG}
Under the assumptions of Theorem \ref{t1}, for any $\epsilon > 0$
there exist constants $C_1$, $C_2$ such that for
$\eta = N^{-2/3 - \epsilon}$
we have \be\label{114}\;\;\;\; \lim_{N\to \infty } \max_{\alpha \leq \varphi^\rho} \max_{i,j}
\hbb{
\E^{\f u} \, \theta \left( N \bar u_\alpha(i) u_\alpha (j) \right)
 -
\E^{\f u} \,\theta \left[\frac N\pi \int_{I} \wt G_{ij}(E + \ii \eta) \, \ind{\lambda_{\alpha - 1} \leq E^- \leq 
\lambda_\alpha} \, \dd E \right]} \;=\; 0\,,
\ee
where
\be\label{defEI}
E^\pm \;\deq\; E\pm\varphi^{C_1}\eta,\;\;\;\;\;\;  I \;\deq\; \qB{-2 - N^{-2/3} \varphi^{C_2} \,,\,  -2 + 
N^{-2/3} \varphi^{C_2}}
\ee
and we introduce the convention $\lambda_0 \deq - \infty$. Here $\f u$ stands for either $\f v$ or $\f w$.
\end{lemma}

\begin{proof}
We shall fix $i,j$ and $\alpha \leq \varphi^\rho$; it is easy to check that all constants in the following are uniform 
in $i,j$, and $\alpha \leq \varphi^\rho$. We write
\be
  \bar u_\alpha(i) u_\alpha(j) \;=\; \frac\eta\pi \,  \int\frac{\bar u_\alpha(i) 
u_\alpha(j)}{(E-\lambda_\alpha)^2+\eta^2} \, \dd E\,.
\ee
Using Theorem \ref{DTH} it is easy to prove that for $C_1$ large enough we have
\be
 \bar u_\alpha(i) u_\alpha(j) \;=\; \frac\eta\pi \,
	\int_a^b \frac{\bar u_\alpha(i) u_\alpha(j)}{(E-\lambda_\alpha)^2+\eta^2}
\, \dd E + O\pbb{\frac{1}{N \varphi^{C_1/2}}}
\ee
 holds \hp for some $c > 0$, as long as
\be
a \;\leq\; \lambda_\alpha^-,\;\;\;\;\;\; b \;\geq\; \lambda_\alpha^+\,,
\ee
where we use the notation \eqref{defEI}, i.e.\ $\lambda_\alpha^\pm \deq \lambda_\alpha \pm \varphi^{C_1} \eta$.
We now choose
\begin{align*}
a \;\deq\; \min \h{\lambda_\alpha^- , \lambda_{\alpha - 1}^+}\,, \qquad
b \;\deq\; \lambda_\alpha^+\,.
\end{align*}

By the assumption on $\theta$ and using Theorem \ref{DTH}, we therefore find
\be\label{19}
\E^{\f u} \,\theta \pb{N \bar u_\alpha(i) u_\alpha(j)}
\;=\;
\E^{\f u}\,\theta\left(\frac{N\eta}\pi \,  \int_a^b \frac{\bar u_\alpha(i) u_\alpha(j)}{(E-\lambda_\alpha)^2+\eta^2} \dd 
E
\right)+o(1)\,.
\ee
Now we split
\begin{align*}
\int_a^b \dd E\;=\; \int_{\lambda_{\alpha - 1}^+}^{\lambda_\alpha^+} \dd E + \ind{\lambda_{\alpha - 1}^+ > 
\lambda_\alpha^-}\int_{\lambda_\alpha^-}^{\lambda_{\alpha-1}^+} \, \dd E
\end{align*}
to get
\begin{equation} \label{lr small}
\E^{\f u} \,\theta \pb{N \bar u_\alpha(i) u_\alpha(j)}
\;=\;
\E^{\f u} \,\theta\left(\frac{N\eta}\pi \int_{\lambda_{\alpha - 1}^+}^{\lambda_\alpha^+} \frac{\bar u_\alpha(i) 
u_\alpha(j)}{(E-\lambda_\alpha)^2+\eta^2} \, \dd E
\right)
+
O\pb{\varphi^{C_0}\,
\E^{\f u} \ind{\lambda_{\alpha - 1}^+ > \lambda_\alpha^-}} +o(1)
\end{equation}
for some constant $C_0$,
where we used Theorem \ref{DTH} and the assumption on $\theta$. Now the level repulsion estimate \eqref{lr1} implies 
that the second term of \eqref{lr small} is $o(1)$. We now observe that, by \eqref{rigidity}, we have $\lambda_\alpha^+ 
\leq -2 + N^{-2/3} \varphi^{C_2}$ and $\lambda_{\alpha - 1}^+ \geq -2 - N^{-2/3} \varphi^{C_2}$ \hp.  It therefore 
easy to see that
\begin{align} \label{uu cutoff}
\E^{\f u} \,\theta \pb{N \bar u_\alpha(i) u_\alpha(j)}
\;=\;
\E^{\f u}\,\theta\left(\frac{N\eta}\pi \int_{I} \frac{\bar u_\alpha(i) u_\alpha(j)}{(E-\lambda_\alpha)^2+\eta^2} \, 
\ind{\lambda_{\alpha - 1} \leq E^- \leq \lambda_\alpha} \, \dd E
\right) + o(1)\,.
\end{align}

Next, we replace the integrand in \eqref{uu cutoff} by $\wt G_{ij}(E+ \ii \eta)$. By definition, we have
\be \label{decomposition of tilde G}
\frac{1}{\eta} \wt G_{ij}(E+ \ii \eta)
\;=\;
\sum_{\beta \neq \alpha} \frac{\bar u_\beta(i) u_\beta(j)}{(E-\lambda_\beta)^2+\eta^2}
+\frac{\bar u_\alpha(i) u_\alpha(j)}{(E-\lambda_\alpha)^2+\eta^2}\,.
\ee
In order to be able to apply the mean value theorem to $\theta$ with the decomposition \eqref{decomposition of tilde G}, 
we need an upper bound on
\begin{align} \label{estimate on tilde G}
\sum_\beta
\frac{N\eta}\pi \int_{I} \frac{\abs{\bar u_\beta(i) u_\beta(j)}}{(E-\lambda_\beta)^2+\eta^2}
\, \dd E \;\leq\; \varphi^{C_0 + C_3} + \varphi^{C_0} \sum_{\beta \geq \varphi^{C_3}} \int_{I} 
\frac{\eta}{(E-\lambda_\beta)^2+\eta^2}
\, \dd E\,,
\end{align}
where the inequality holds \hp for any $C_3$; here we used Theorem \ref{DTH}. Using $\gamma_\beta \geq -2 + c 
(\beta/N)^{2/3}$ as well as \eqref{rigidity}, we find \hp for large enough $C_3$
\begin{align} \label{tail estimate}
\varphi^{C_0} \sum_{\beta \geq \varphi^{C_3}} \int_{I} \frac{\eta}{(E-\lambda_\beta)^2+\eta^2}
\, \dd E \;\leq\; \varphi^{C_0 + C_2} N^{-2/3} \sum_{\beta \geq \varphi^{C_3}} \frac{\eta}{(\beta/N)^{4/3}} \;\leq\; 
N^{- \epsilon/2}\,.
\end{align}
Thus the left-hand side of \eqref{estimate on tilde G} is bounded by $\varphi^{C_0 + C_3 + 1}$.

Let us abbreviate $\chi(E) \deq \ind{\lambda_{\alpha - 1} \leq E^- \leq \lambda_\alpha}$.
Now, recalling the assumption on $\theta$, we may apply the mean value theorem as well as Theorem \ref{DTH} to get
\begin{multline}
\left|
\E^{\f u}\,\theta\left(\frac{N\eta}\pi \int_{I} \frac{\bar u_\alpha(i) u_\alpha(j)}{(E-\lambda_\alpha)^2+\eta^2}
\, \chi(E) \, \dd E
\right)
-
\E^{\f u}\,\theta\left(\frac{N }\pi \int_{I} \wt G_{ij}(E + \ii \eta)
\, \chi(E) \, \dd E
\right)
\right|
\\\label{1119} \leq\; \varphi^{\wt C}
\;\E^{\f u}\, \sum_{\beta \neq \alpha} \frac{N\eta}\pi \, \int_{I} \frac{\abs{\bar u_\beta(i) 
u_\beta(j)}}{(E-\lambda_\beta)^2+\eta^2}\, \chi(E) \, \dd E
\end{multline}
for some constant $\wt C \leq C (C_0 + C_3 + 1)$ independent of $C_1$.
We now estimate the right-hand side of \eqref{1119}. Exactly as in \eqref{tail estimate}, one finds that there exists 
$C_4$ such that the contribution of $\beta \geq \varphi^{C_4}$ to the right-hand side of \eqref{1119} vanishes in the 
limit $N \to \infty$.  Next, we deal with the eigenvalues $\beta < \alpha$ (in the case $\alpha > 1$).  Using Theorem 
\ref{DTH} we get
\begin{align*}
\sum_{\beta < \alpha} \frac{N\eta}\pi \, \E^{\f u} \int_{I} \frac{\abs{\bar u_\beta(i) 
u_\beta(j)}}{(E-\lambda_\beta)^2+\eta^2}\, \chi(E) \, \dd E
\;\leq\; \varphi^{C} \, \E^{\f u} \int_{\lambda_{\alpha - 1}^+}^\infty \frac{\eta}{(E-\lambda_{\alpha - 1})^2+\eta^2} \, 
\dd E \;\leq\; \varphi^{-\wt C -c}\,,
\end{align*}
where $c>0$ for $C_1$ large enough.

What remains is the estimate of the terms $\alpha < \beta \leq \varphi^{C_4}$ in \eqref{1119}.
For a given constant $C_5>0$ we partition $I = I_1 \cup I_2$ with $I_1 \cap I_2 = \emptyset$ and
\be
I_1 \;\deq\; \Big\{E\in I \st \exists \, \beta \,,\, \alpha < \beta \leq  \varphi^{  C_4} \,,\, 
|E-\lambda_\beta|\leq \eta \varphi^{  C_5}
	 \Big\}\,.
\ee
It is easy to see that, for large enough $C_5$, we have
\begin{align*}
\sum_{\beta \st \alpha < \beta \leq \varphi^{C_4}} \frac{N\eta}\pi \, \E^{\f u} \int_{I_2} \frac{\abs{\bar u_\beta(i) 
u_\beta(j)}}{(E-\lambda_\beta)^2+\eta^2}\, \chi(E) \, \dd E \;\leq\; \varphi^{-\wt C - c}
\end{align*}
where $c > 0$. Let us therefore consider the integral over $I_{1}$. One readily finds, for $\lambda_\alpha \leq 
\lambda_{\alpha +1} \leq \lambda_\beta$, that
\be\label{HNK}
\frac{  1\, }{(E-\lambda_\beta)^2+\eta^2} \, \ind{E^- \leq \lambda_\alpha}
\;\leq\;
 \frac{\varphi^{2C_1}}{(\lambda_\beta-\lambda_\alpha)^2+\eta^2}
 \;\leq\; \frac{\varphi^{2C_1}}{(\lambda_{\alpha + 1}-\lambda_\alpha)^2+\eta^2}\,.
\ee
From Theorem \ref{DTH} we therefore find that there exists a constant $C_6$, depending on $C_1$, such that
\begin{align} \label{118}
\sum_{\beta \st \alpha < \beta \leq \varphi^{C_4}} \frac{N\eta}\pi \, \E^{\f u} \int_{I_1} \frac{\abs{\bar u_\beta(i) 
u_\beta(j)}}{(E-\lambda_\beta)^2+\eta^2}\, \chi(E) \, \dd E \;\leq\; \varphi^{C_6} \, \E^{\f u} \,
\frac{\eta^2}{(\lambda_{\alpha + 1}-\lambda_\alpha)^2+\eta^2}
\end{align}
The right-hand side of \eqref{118} is bounded by $\E^{\f u}{\bf 1}(|\lambda_{\alpha + 1}-\lambda_\alpha|\leq  
N^{-1/3}\eta^{1/2} ) +O(N^{-\e})$.  Using \eqref{lr1} we now obtain
\be\label{119}
\sum_{\beta \st \alpha < \beta \leq \varphi^{C_4}} \frac{N\eta}\pi \, \E^{\f u} \int_{I_1} \frac{\abs{\bar u_\beta(i) 
u_\beta(j)}}{(E-\lambda_\beta)^2+\eta^2}\, \chi(E) \, \dd E \;\leq\; \varphi^{- \wt C - c}
\ee
where $c > 0$. This concludes the proof.
\end{proof}

In a second step we convert the cutoff function in lemma \ref{CTG}   into a function of $\wt G_{ij}$. 

\begin{lemma}\label{GCC}
Recall the definition \eqref{thetam} of the approximate delta function $\theta_\eta$ on the scale $ \eta$.  Let $\alpha 
\leq \varphi^\rho$ and $q \equiv q_\alpha:\R \to\R_+$ be a smooth cutoff function concentrated around $\alpha - 1$, 
satisfying
\[
q(x) = 1 \quad  \text{if} \quad |x - \alpha + 1| \le 1/3,   \qquad q(x) = 0   \quad  \text{if} \quad |x - \alpha + 1| 
\ge 2/3\,.
\]
 Let 
\be
\chi \;\deq\; {\bf 1}_{[E_L, E^-]}\,, \qquad E_L \;\deq\; -2 - 2N^{-2/3}(\log N)^L
\ee
where
\be \label{def eta wt eta}
\eta \;\deq\; N^{-2/3-\e}\,, \qquad \tilde\eta \;\deq\; N^{-2/3-6\e}
\ee
for $\epsilon > 0$. Then for $\epsilon$ small enough we have (recall the definition \eqref{defEI})
\be\label{115}
\lim_{N\to \infty } \max_{\alpha \leq \varphi^\rho} \max_{i,j}
\hbb{
\E^{\f u} \, \theta \pb{N \bar u_\alpha(i) u_\alpha(j)}
 -
\E^{\f u} \,\theta \left[\frac N\pi \, \int_{I } \wt G_{ij}(E + \ii \eta) \, q \qb{ \tr  (\chi  \ast \theta_{\tilde 
\eta}) (H)} \, \dd E\right] } \;=\; 0\,.
\ee
Here $\f u$ stands for either $\f v$ or $\f w$.
\end{lemma}

\begin{proof}
Note first that
\begin{align*}
\frac N\pi \int_{I} \wt G_{ij}(E + \ii \eta) \, \ind{\lambda_{\alpha - 1} \leq E^- \leq \lambda_\alpha} \, \dd E
&\;=\; \frac N\pi \int_{I} \wt G_{ij}(E + \ii \eta) \, \indb{\cal N(-\infty, E^-) = \alpha - 1} \, \dd E
\\
&\;=\; \frac N\pi \int_{I} \wt G_{ij}(E + \ii \eta) \, q \qb{\tr \chi(H)} \, \dd E
\end{align*}
\hp.

Next, recall that \eqref{6.10} asserts that for $\ell=N^{-2/3-2\e}$ we have
\be\label{6.101}
\absb{\tr \chi  (H) - \tr  (\chi   \ast \theta_{\tilde \eta}) (H)}  \;\le\; C\left( N^{-\e}  +   \cN (E^- -\ell , E^- 
+\ell )  \right)
\ee
\hp for sufficiently large $N$. We therefore find that
\begin{align*}
&\qquad \left|\frac N\pi \, \int_{I } \wt G_{ij}(E + \ii \eta) \, \indb{\lambda_{\alpha - 1} \leq E^- \leq 
\lambda_\alpha} \, \dd E -
\frac N\pi\int_{I} \wt G_{ij}(E + \ii \eta)
\, q \qb{ \tr (\chi  \ast \theta_{\tilde \eta}) (H)} \, \dd E
\right|
\\
&\leq\;
CN\sum_{\beta=1}^N\int_{I} \absb{\wt G_{ij}(E + \ii \eta)}
{\bf 1}(|E^- - \lambda_\beta|\leq  \ell )
\, \dd E + \varphi^C N^{-\epsilon}
\\
&\leq\;
CN \sum_{\beta=1}^{\varphi^C}\int_{I} \absb{\wt G_{ij}(E + \ii \eta)} \, {\bf 1}(|E^- - \lambda_\beta|\leq  \ell ) \, 
\dd E + \varphi^C N^{-\epsilon}
\\
&\leq\; C \varphi^C N \ell \, \sup_{E \in I} \abs{\wt G_{ij}(E + \ii \eta)}
+
\varphi^C N^{-\epsilon}
\end{align*}
holds \hp, where in the first inequality we estimated the integral $\int_I \abs{\wt G_{ij}(E + \ii \eta)} \, \dd E$ 
exactly as \eqref{estimate on tilde G}, and in the second inequality we used \eqref{rigidity}. Using the definition of 
$I$ and \eqref{Lambdaodfinal} we get
\begin{align*}
\sup_{E \in I} \abs{\wt G_{ij}(E + \ii \eta)} \;\leq\; \varphi^C \pB{N^{-1/3} + N^{-2/3} \eta^{-1/2} + N^{-1} \eta^{-1}} 
\;\leq\; N^{-1/3 + \epsilon}\,.
\end{align*}
Together with \eqref{114}, the claim follows.
\end{proof}

In a third and final step, we use the Green function comparison method to show the following statement.

\begin{lemma} \label{33} Under the assumptions of Lemma \ref{GCC}, we have
\begin{equation*}
\lim_{N \to \infty} \max_{i,j} \; \pb{\E^{\f v} - \E^{\f w}} \, \theta \qbb{\frac N\pi \int_{I } \wt G_{ij}(E + \ii \eta) 
\, q\qb{ \tr  (\chi  \ast \theta_{\tilde \eta}) (H)} \, \dd E}
\;=\; 0\,.
\end{equation*}
\end{lemma}
The rest of this section is devoted to the proof of Lemma \ref{33}.

\subsection{Green function comparison: proof of Lemma \ref{33}} \label{pf of 33}
The claimed uniformity in $i$ and $j$ is easy to check in our proof, and we shall not mention it anymore. Throughout the 
following we rename $i = \alpha$ and $j = \beta$ in order to use $i$ and $j$ as summation indices. We now fix $\alpha$ 
and $\beta$ for the whole proof. (Note that $\alpha$ and $\beta$ need not be different.)

We use the identity (see \eqref{def of tilde G})
\begin{equation} \label{tilde G sum}
\wt G_{ij}(z) \;=\;  (\im z) \sum_k X_{ij,k}(z) \,, \qquad  X_{ij,k}(z)  \;\deq\; G_{ik}(z)\overline{ G_{jk}(z)}\,.
\end{equation}
We begin by dropping the diagonal terms in \eqref{tilde G sum}.
\begin{lemma} \label{dropping diagonal terms}
For small enough $\epsilon > 0$ we have
\begin{align} \label{no diag terms}
\E^{\f u} \, \theta \qbb{\frac N\pi \int_{I } \wt G_{\alpha \beta}(E + \ii \eta) \, q\qb{ \tr  (\chi  \ast \theta_{\tilde 
\eta}) (H)} \, \dd E}
-
\E^{\f u} \, \theta \qbb{\int_{I } x(E) \, q(y(E)) \, \dd E}
\;=\; o(1)\,,
\end{align}
where $\f u$ stands for either $\f v$ or $\f w$, and
\be\label{defxy}
x(E) \;\deq\; \frac {N \eta} \pi\sum_{k\neq \al,\beta} X_{\al\beta, k}(E + \ii \eta)\,,
\qquad
y(E) \;\deq\; \tilde \eta \int_{E_L}^{E^-}  \sum_{i \neq k}X_{ii,k}(\tilde E + \ii \tilde \eta) \, \dd \tilde E\,.
\ee
\end{lemma}
\begin{proof}
We estimate
\begin{align*}
\absbb{\frac{N}{\pi}\wt G_{\alpha \beta}(E + \ii \eta) - x(E)} \;\leq\; \varphi^C N \eta \;\leq\; \varphi^C N^{1/3 - 
\epsilon}
\end{align*}
\hp and, recalling that $q'$ is bounded,
\begin{align*}
\absB{q\qb{ \tr  (\chi  \ast \theta_{\tilde \eta}) (H)} - y(E)} \;\leq\; \varphi^C \tilde \eta N N^{-2/3} \;\leq\; 
\varphi^C N^{-1/3 - 6 \epsilon}
\end{align*}
\hp. Therefore the difference of the arguments of $\theta$ in \eqref{no diag terms} is bounded by $\varphi^C N^{-1/3 - 
\epsilon}$ \hp. (Recall that $\abs{I} \leq \varphi^C N^{-2/3}$.) Moreover, since $q$ is bounded, it is easy to see that 
both arguments of $\theta$ in \eqref{no diag terms} are bounded \hp by
\begin{align*}
\varphi^C N \eta N^{-2/3} \pb{1 + N \sup_{E \in I} \Lambda_0^2(E + \ii \eta)} \;\leq\; \varphi^C N^\epsilon\,,
\end{align*}
where we used Theorem \ref{45-1}. The claim now follows from the mean value theorem and the assumption on $\theta$.
\end{proof}

For the following we work on the product probability space of the ensembles $H^\bv$ and $H^\bw$. To distinguish them we 
denote the elements of $H^\bv$ by $N^{1/2} v_{ij}$ and the elements of $H^\bw$ by $N^{-1/2} w_{ij}$. We fix a bijective 
ordering map $\Phi$ on the index set of the independent matrix elements,
\[
\Phi \st \{(i, j) \st 1\le i\le  j \le N \} \;\to\; \h{1, \ldots, \gamma_{\rm max}} \,, \qquad \gamma_{\rm max} \;\deq\; 
\frac{N(N+1)}{2}\,,
\]
and denote by  $H_\gamma$ the generalized Wigner matrix whose matrix elements $h_{ij}$ follow
the $v$-distribution if $\Phi(i,j) \le \gamma$ and the $w$-distribution
otherwise. In particular, $H_0 = H^{\f v}$ and $ H_{\gamma_{\rm max}} = H^{\f w}$. Hence
\begin{align}\label{tel}
& \big [ \E^{\f v} - \E^{\f w} \big ]  \,\theta \left[\int_{I } x(E) \, q (y(E)) \, \dd E \right] \;=\; \sum_{\gamma = 
1}^{\gamma_{\rm max}} \Big [ \E^{(H_{\gamma-1})}  - \E^{(H_{\gamma})} \Big ]  \,\theta \left[  \int_{I } x(E) \, q 
(y(E)) \, \dd E\right]
\end{align}
(in self-explanatory notation).

Let us now fix a $\gamma$ and let $(a,b)$ be determined by $\Phi (a, b) = \gamma$. Throughout the following we consider 
$\alpha, \beta, a, b$ to be arbitrary but fixed and often omit dependence on them from the notation. Our strategy is to 
compare $H_{\gamma-1}$ with $H_\gamma$ for each $\gamma$. In the end we shall sum up the differences in the telescopic 
sum \eqref{tel}.

Note that $H_{\gamma - 1}$ and $H_\gamma$ differ only in the matrix elements indexed by $(a,b)$ and $(b,a)$.  Let 
$E^{(ij)}$ denote the matrix whose matrix elements are zero everywhere except
at position $(i,j)$ where it is 1; in other words, $E^{(ij)}_{k\ell}=\delta_{ik}\delta_{j\ell}$.
Thus we have
\begin{align}
H_{\gamma-1} &\;=\; Q + \frac{1}{\sqrt{N}}V\,, \qquad V \;\deq\; v_{ab}E^{(ab)}
+ v_{ba}  E^{(ba)}\,,
\notag \\ \label{defHg1}
H_\gamma &\;=\; Q + \frac{1}{\sqrt{N}} W\,, \qquad W \;\deq\; w_{ab}E^{(ab)} +
	w_{ba} E^{(ba)}\,.
\end{align}
Here $Q$ is the matrix obtained from $H_{\gamma}$ (or, equivalently, from $H_{\gamma - 1}$) by setting the matrix 
elements indexed by $(a,b)$ and $(b,a)$ to zero.
Next, we define the Green functions
\be\label{defG}
		R \;\deq\; \frac{1}{Q-z}\,, \qquad S \;\deq\; \frac{1}{H_{\gamma-1}-z}\,.
\ee
We shall show that the difference between the expectation $\E^{(H_{\gamma-1})}$ and
 $\E^{(Q)}$ depends only on the second moments of $v_{ab}$, up to an error term that is affordable even after summation 
over $\gamma$. Together with same argument applied to $\E^{(H_\gamma)}$, and the fact that the second moments of 
$v_{ab}$ and $w_{ab}$ are identical, this will prove Lemma \ref{33}.

For the estimates we need the following basic result, proved in \cite{EYYrigi} (Equation (6.32)).
\begin{lemma} \label{lemma: Rij bound}
For any $\eta' \deq N^{-2/3 - \delta}$ we have \hp
\begin{align*}
\sup_{E \leq N^{-2/3 + \epsilon}}\max_{i,j} \absb{R_{ij}(E + \ii \eta') - \delta_{ij}m_{sc}(E + \ii \eta')} \;\leq\; 
\Lambda_\delta \;\deq\; N^{-1/3 + 2 \delta}\,.
\end{align*}
The same estimates hold for $S$ instead of $R$.
\end{lemma}

Our comparison is based on the resolvent expansion
\be\label{SR-N}
	 S \;=\;  R -  N^{-1/2} RVR+  N^{-1} (RV)^2R -  N^{-3/2} (RV)^3R+  N^{-2} (RV)^4S.
\ee
Using Lemma \ref{lemma: Rij bound} we easily get \hp, for $i \neq j$,
\be
| S _{ij}-R_{ij}| \;\leq\; \varphi^C N^{-1/2}\Lambda_\epsilon^{2-r}
\qquad \text{where }
r \;\deq\; {\bf 1}(i\in\{a,b\})+{\bf 1}(j\in\{a,b\})\,.
\ee
Defining
\begin{align} \label{def: Delta X}
\Delta X_{ij,k} \;\deq\; S_{ik}\overline{S_{jk}}-R_{ik}\overline{R_{jk}}\,,
\end{align}
we therefore have the trivial bound \hp
\be\label{346}
\left|\Delta X_{ij,k} \right| \;\leq\; \varphi^C N^{-1/2}\Lambda_\epsilon^{3-s} \qquad (k \neq i,j)\,,
\ee
where we abbreviated
\begin{align}
s \;\deq&\; \max\hB{{\bf 1}(i\in\{a,b\})+{\bf 1}(k\in\{a,b\})\,,\, {\bf 1}(j\in\{a,b\})+{\bf 1}(k\in\{a,b\})}
\notag \\ \label{defsmin}
\;=&\; \indb{\{i,j\}\cap\{a,b\} \neq \emptyset}+ {\bf 1}(k\in\{a,b\})\,.
\end{align}
The variable $s$ counts the maximum number of diagonal resolvent matrix elements in $\Delta X_{ij,k}$. The bookkeeping 
of $s$ will play a crucial role in our proof, since the smallness associated with off-diagonal elements (see Lemma 
\ref{lemma: Rij bound}) is needed to control the resolvent expansion \eqref{SR-N} under the two-moment matching 
assumption.

From now on it is convenient to modify slightly our notation, and to write $\E F(H_\gamma)$ instead of $\E^{(H_\gamma)} 
F(H)$. We also use $\E_{ab}$ to denote partial expectation obtained by integrating out the variables $v_{ab}$ and 
$w_{ab}$.

By applying \eqref{SR-N} to \eqref{def: Delta X} and taking the partial expectation $\E_{ab}$, one finds, as above, that 
there exists a random variable $A_1$, which depends on the randomness only through $Q$ and the first two moments of 
$v_{ab}$, such that for $k\neq i,j$ and $s$ as in\eqref{defsmin} we have \hp
\be\label{338}
\left|\E_{ab} \, \Delta X_{ij,k} -A_1\right|
\;\leq\;
\varphi^C N^{-3/2}\Lambda_\epsilon^{3-s }\,.
\ee
Using this bound we may estimate
\be
\Delta x(E) \;\deq\; x^S(E) - x^{R}(E) \,, \qquad \Delta y(E) \;\deq\; y^S(E) - y^R(E)\,,
\ee
with the convention that a superscript $S$ denotes a quantity defined in terms of the matrix $H_{\gamma - 1}$, and a 
superscript $R$ a quantity defined in terms of the matrix $Q$.

\begin{lemma} \label{lemma: main comparison estimate}
For fixed $\al,\beta,a,b$ there exists exists a random variable $A$, which depends on the randomness only through $Q$ 
and the first two moments of $v_{ab}$, such that
\be\label{340}
\E \, \theta \left[  \int_{I }  x^S(E) \, q\pb{ y^S(E)} \,\dd E\right]
- \E \, \theta \left[   \int_{I }  x^R(E) \, q\pb{ y^R(E)} \, \dd E\right] \;=\; A + o\pb{N^{-2 + t}+N^{-2 + 
{\bf 1}(a=b)}}\,,
\ee
where $t \deq {|\{a,b\}\cap   \{\al,\beta\}|}$\,.
\end{lemma}

Before proving Lemma \ref{lemma: main comparison estimate}, we show how it implies Lemma \ref{33}.
\begin{proof}[Proof of Lemma \ref{33}]
It suffices to prove that each summand in \eqref{tel} is bounded by $o\pb{N^{-2 + t} + N^{-2 + \ind{a = b}}}$. This 
follows immediately by applying Lemma \ref{lemma: main comparison estimate} to $S = (H_{\gamma - 1} - z)^{-1}$ and $S' \deq (H_\gamma - z)^{-1}$ and 
subtracting the statements; note that the random variables $A$ in the statement of Lemma \ref{lemma: main comparison 
estimate} are by definition the same for $S$ and $S'$.
\end{proof}

\begin{proof}[Proof of Lemma \ref{lemma: main comparison estimate}]
Throughout the proof of Lemma \ref{lemma: main comparison estimate} we shall abbreviate $H \equiv H_{\gamma - 1} = 
(h_{ij})$, as well as $S \equiv S(z) = (H - z)^{-1}$.

Since $E\in I$ (recall \eqref{defEI}) we get from Theorem \ref{45-1} that \hp
\be \label{estimate on x(E)}
|x(E)| \;\leq\; \varphi^C N^2\eta \Lambda^2_\epsilon \;\leq\; \frac{N^{C \epsilon}}{
\eta }\,,
\ee
which implies
\be \label{bound on x}
 \int_{I } | x(E) | \,
 \dd E \;\leq\; N^{C \e}\,.
\ee
Here we adopt the convention that if $x$ or $y$ appears without a superscript, the claim holds for both superscripts $R$ 
and $S$. Similarly, we find \hp
\be \label{bound on y}
|y(E)| \;\leq\; \tilde \eta \, \varphi^C N^{-2/3}N^2 \Lambda^2_{6 \epsilon} \;\leq\; N^{C \e}\,.
\ee

Next, in the definition of $x(E)$ and $y(E)$ we condition on the variable $s$ defined in \eqref{defsmin} by introducing, 
for $s = 0,1,2$,
\begin{align*}
x_s(E) &\;\deq\; \frac {N \eta} \pi\sum_{k\neq \al,\beta} X_{\al\beta, k}(E + \ii \eta)
\, \indB{s=\indb{\{\al, \beta\} \cap \{a,b\} \neq \emptyset} + {\bf 1}(k\in\{a,b\})}\,,
\\
y_s(E) &\;\deq\; \tilde \eta \int_{E_L}^{E^-}  \sum_{i \neq k}X_{ii,k}(\tilde E + \ii \tilde \eta) \, \dd \tilde E \, 
\indB{s={\bf 1} (i\in\{a,b\})  + {\bf 1}(k\in\{a,b\})}\,.
\end{align*}
As above, $s$ is a bookkeeping index that bounds the number of diagonal resolvent matrix elements appearing in the 
resolvent expansion.

We abbreviate $\Delta x_s(E) \deq x_s^S(E)-x_s^R(E)$ and $\Delta y_s(E)= y_s^S(E)-y_s^R(E)$. Recalling 
the definition $t = {|\{a,b\}\cap  \{\al,\beta\}|}$, we find \hp
\be\label{deftab}
\left|\Delta x_s(E)\right| \;\leq\; \varphi^C N \eta N^{-1/2}\Lambda_\epsilon^{3-s} N^{\ind{s = \ind{t > 0}}} \;\leq\; 
\frac{\eta^{s-2}}{N^{ 3/2-t - C \epsilon} }\,,
\ee
where we used Theorem \ref{45-1} and the elementary inequality $s + \indb{s = \ind{t > 0}} \leq t + 1$ which holds if 
$x_s(E) \neq 0$.
Thus we get \hp
\be\label{346a}
  \int_{I }   |\Delta x_s(E)|
 \, \dd E \;\leq\;   \frac{\eta^{s-1}}{N^{ 3/2-t-C\e} } \;=\; N^{-5/6} N^{-2s/3 + t+C\e}\,.
 \ee
Now we may argue similarly to \eqref{338}. We find that, for any $E$-dependent random variable $f \equiv f(E)$ 
independent of $h_{ab}$, there exists a random variable $A_2$, which depends on the randomness only through $Q$, $f$, 
and the first two moments of $h_{ab}$, such that \hp
\be\label{346b}
 \left| \int_{I }  \left( \E_{ab} \, \Delta x_s(E) \right) f(E)
 \, \dd E-A_2\right| \ind{\Omega}
 \;\leq\;   \|f \, \ind{\Omega}\|_\infty \, N^{-11/6} \, N^{-2s/3+t+C\e}\,,
\ee
where $\Omega$ is any event. Note that, as in \eqref{338}, we find that \eqref{346b} is suppressed by a factor $N^{-1}$ 
compared to \eqref{346}. This may be easily understood, as the leading order error term in the resolvent expansion of 
\eqref{346} is of order $1$ in $H$, whereas the leading order error term in \eqref{346b} is of order $3$ in $H$. These 
error terms have the same number of off-diagonal elements (estimated using Lemma \ref{lemma: Rij bound}), and the same 
entropy factor of the summation indices.

We may derive similar bounds for $y_s(E)$. As in \eqref{346}, we have \hp
\be\label{344a}
\left|\Delta y_s(E)\right| \;\leq\; \varphi^C \tilde \eta \, N^{-2/3}N^{2-s} N^{-1/2} \Lambda^{3-s}_{6 \epsilon} 
\;\leq\; N^{-5/6} N^{-2s/3+C\e}\,.
\ee
Furthermore, we find that there exists an $E$-dependent random variable $A_3(E)$, which depends on the randomness only 
through $Q$ and the first two moments of $h_{ab}$, such that \hp
\be\label{344b}
\left|\E_{ab}^{(H_{\gamma - 1})} \Delta y_s(E)-A_3(E)\right| \;\leq\;    N^{-11/6} N^{-2s/3+C\e}\,.
\ee 
 
After these preparations, we may now estimate the error resulting from setting $h_{ab}$ to zero in the expression $\E\, 
\theta \left[  \int_{I }  x(E)\, q(  y(E)) \, \dd E\right]$. Recalling the conditioning over $s = 0,1,2$, we find
\begin{equation*}
\theta \left[   \int_{I }  x^S\, q(  y^S) \, \dd E\right]
\\
=\; \theta \left[  \int_{I }  \left(x^R+\Delta x_0+\Delta x_1+\Delta x_2\right)\, q \pB{y^R+\Delta y_0+\Delta 
y_1+\Delta y_2} \, \dd E\right]\,;
\end{equation*}
here and in the following we omit the argument $E$ unless it is needed.
Using \eqref{344a} we have \hp
\begin{multline*}
\theta \left[ \int_{I }  \left(x^R+\Delta x_0+\Delta x_1+\Delta x_2\right)\, q \pB{y^R+\Delta y_0+\Delta y_1 + 
\Delta y_2} \, \dd E\right]
\\
=\; \theta \left[  \int_{I } \pB{x^R+\Delta x_0+\Delta x_1+\Delta x_2}\, \pB{q ( y^R)+ q'(y^R) \left(\Delta y_0 + \Delta 
y_1   \right) + q''(y^R) (\Delta y_0)^2} \,  \dd E\right]+o(N^{-2})\,.
\end{multline*}
The use of the mean value theorem for $\epsilon$ small enough is easy to justify using the assumption on $\theta$ and 
the bounds \eqref{bound on x} and \eqref{bound on y}. In the following we shall no longer mention such estimates of the 
argument of derivatives of $\theta$, which can always be easily checked in a similar fashion.

Recall that an error of order $o(N^{-2+t})$ is affordable in the error estimate. Thus, using the basic power counting 
given by \eqref{bound on x}, \eqref{bound on y}, \eqref{346a}, and \eqref{344a}, we find \hp
\begin{multline} \label{354}
\theta \left[  \int_{I }  x^S\, q(  y^S)  \, \dd E \right]
-\theta \left[   \int_{I }  x^R\, q(  y^R)  \, \dd E \right]
\;=\; \theta' \left[   \int_{I }  x^R\, q(  y^R)  \, \dd E \right]
\\
{}\times{}\left[  \int_{I }\Big( \left(
 \Delta x_0+\Delta x_1 \right) q ( y^R) + x^Rq '( y^R) \left(\Delta y_0+\Delta y_1  \right) +  \Delta x_0 \, 
q '( y^R)\Delta y_0
 +x^R q''(y^R) (\Delta y_0)^2  \Big) \, \dd E \right]\\
 +\frac12 \, \theta'' \left[   \int_{I }  x^R\, q(  y^R)  \, \dd E \right] \left[  \int_{I }\Big(  \Delta x_0 \, 
q ( y^R) + x^Rq '( y^R) \Delta y_0 \Big) \, \dd E \right] ^2 + o(N^{-2 + t})\,.
\end{multline}
We now start dealing with the individual terms on the right-hand side of \eqref{354}.

First, we consider the terms containing $\Delta x_1$ and $\Delta y_1$. Applying \eqref{346b} and \eqref{344b} we find 
that there exists a random variable $A_4$, which depends on the randomness only through $Q$ and the first two moments of 
$h_{ab}$, such that
\be \left|\E_{ab} \,    \int_{I }\left( \Delta x_1 \, q ( y^R)+x^Rq '( y^R) \Delta y_1\right)  \, \dd E  -A_4\right| 
\;=\; o(N^{-2 + t})
 \ee
\hp.
Inserting this into \eqref{354}, we find \hp
\begin{multline}\label{356}
\E_{ab}\, \theta \left[  \int_{I }  x^S\, q(  y^S)  \, \dd E \right]
-\E_{ab}\,\theta \left[   \int_{I }  x^R\, q(  y^R)  \, \dd E \right]
\;=\; \E_{ab}\,\theta' \left[   \int_{I }  x^R\, q(  y^R)  \, \dd E \right]
\\
{}\times{}\left[  \int_{I }\Big(\Delta x_0 \, q ( y^R) + x^Rq '( y^R) \Delta y_0 +  \Delta x_0 q '( 
y^R)\Delta y_0
 +x^R q''(y^R) (\Delta y_0)^2  \Big) \, \dd E \right]
 \\
+\frac12 \, \E_{ab}\,\theta'' \left[   \int_{I }  x^R\, q(  y^R)  \, \dd E \right] \left[  \int_{I }\Big(  \Delta x_0 \, 
q ( y^R) + x^Rq '( y^R) \Delta y_0 \Big) \, \dd E \right] ^2
+ A_4+o(N^{-2 + t})\,.
 \end{multline}

Thus we only need to focus on the error terms $\Delta x_0$ and $\Delta y_0$. Note that we have
\begin{align} \label{dx0}
\Delta x_0(E) &\;=\; \ind{t = 0} \, \frac {N \eta} \pi\sum_{k\neq \al,\beta,a,b} \Delta X_{\alpha \beta, k} (E + \ii 
\eta)
\\ \label{dy0}
\Delta y_0(E) &\;=\; \tilde \eta \int_{E_L}^{E^-} \dd \tilde E \, \sum_{i \neq k} \ind{i,k \notin \{a,b\}} \, \Delta 
X_{ii,k}(\tilde E + \ii \tilde \eta)\,.
\end{align}
Recall that the $(i,j)$-component of the resolvent expansion \eqref{SR-N} reads
\be
    S_{ij} = \left( R -  N^{-1/2} RVR+  N^{-1} (RV)^2R -  N^{-3/2} (RV)^3R+  N^{-2} (RV)^4S\right)_{ij}.
\ee
Now we assume that $i\neq j$ and $|\{i,j\} \cap \{a,b\}|=0$. It is easy to see that this assumption holds for any matrix 
element in the formulas \eqref{dx0} and \eqref{dy0}. Then we can use Lemma \ref{lemma: Rij bound} to estimate the $m$-th 
term as follows:
\be\label{51}
\absB{N^{-m/2}\big[(RV)^m R\big]_{ij}} \;\leq\; N^{-m/2+C\e}N^{-2/3} \,, \qquad
\absB{N^{-2}\big[(RV)^4 S\big]_{ij}} \;\leq\;  N^{-8/3+C\e}\,,
\ee
\hp. 

Next, we apply the resolvent expansion to $X_{ij,k}$. Note that in our applications errors of size $O(N^{-8/3-c})$ are 
affordable in $\Delta X_{ij,k}$ for some $c > 0$ independent of $\epsilon$ (see \eqref{no diag terms} and 
\eqref{defxy}). Now let us assume that the indices $i, j, a,b,k$ satisfy the condition
\begin{itemize}
\item[$(*)$]  $\{i,j\}\cap \{a,b\} = \emptyset$ and $k \neq i, j,a,b$.
\end{itemize}
In the applications we shall set $i = \alpha$ and $j = \beta$ in \eqref{dx0}, and $i = j$ in \eqref{dy0}. In both cases, 
it is easy to check that the condition $(*)$ is satisfied for nonvanishing summands.

We can therefore separate $\Delta X_{i j,k}$ into three parts, indexed according to how many $V$-matrix elements they 
contain,
\be \label{Delta X split}
\Delta X_{i j,k} \;=\; \Delta X^{(1)}_{i j,k} +\Delta X^{(2)}_{i j,k} +\Delta X^{(3)}_{i j,k} +O(N^{-3+C\e})
\ee
\hp;
here we defined
\begin{align}
\label{sdd1}
\Delta X^{(1)}_{i j,k} &\;\deq\; - N^{-1/2} R_{i k} \overline{(RVR)}_{j k} + [C]_1\,,
\\
\label{sdd2}
\Delta X^{(2)}_{i j,k} &\;\deq\; N^{-1 }R_{i k}\overline{(RVRVR)}_{j k} + N^{-1 }(RVR) _{i k}\overline{(RVR)}_{j k} + 
[C]_1\,,
\\
\label{sdd3}
\Delta X^{(3)}_{i j,k} &\;\deq\; - N^{-3/2 }R_{i k}\overline{(RVRVRVR)}_{j k}
- N^{-3/2 }(RVR)_{i k}\overline{(RVRVR)}_{j k} + [C]_2\,,
\end{align}
where $[C]_l$, $l = 1,2$, means the complex conjugate of the first $l$ terms on the right-hand side with $i$ and $j$ 
exchanged.
Furthermore, it easy to see that the second term on the right-hand side of \eqref{sdd3} is of order $ O(N^{-17/6+C\e})$.
Thus we find \hp
\begin{align}
-\Delta X^{(3)}_{i j,k} &\;=\; N^{-3/2 } R_{i k}\overline{(RVRVRVR)}_{j k}+N^{-3/2 }(RVRVRVR)_{i k}\overline{R}_{j k} + 
O(N^{-17/6 + C \epsilon})
\notag \\
&\;=\; Y + O(N^{-17/6 + C \epsilon})\,,
\end{align}
where $Y$ is a finite sum of terms of the form
\begin{align} \label{defY}
N^{-3/2 }R_{i k}\overline{(R_{j a}V_{ab}R_{bb}V_{ba}R_{aa}V_{ab}R_{bk})}
\end{align}
and terms obtained from \eqref{defY} by (i) taking the complex conjugate and exchanging $i$ and $j$, and (ii) exchanging 
$a$ and $b$. Using Lemma \ref{lemma: Rij bound} we find that \eqref{defY} is equal to
\begin{equation*}\label{xgamma}
N^{-3/2 }R_{i k}\overline{(R_{j a}V_{ab}R_{bb}V_{ba}R_{aa}V_{ab}R_{bk})}
\;=\; N^{-3/2 } \overline{m_{sc}^2} \, R_{i k} \, \overline{R_{ja} V_{ab} V_{ba} V_{ab} R_{bk}} + O(N^{-17/6 + C 
\epsilon})
\end{equation*}
\hp.
The splitting \eqref{Delta X split} induces a splitting
\be
 \Delta x_0=\Delta x^{(1)}_0+\Delta x^{(2)}_0+\Delta x^{(3)}_0+ O(N^{-5/3+C\e})
\ee
\hp in self-explanatory notation. It is easy to see that
\begin{align} \label{Delta x estimates}
\abs{\Delta x_0^{(1)}} \;\leq\; N^{-1/6 + C \epsilon} \,, \qquad
\abs{\Delta x_0^{(2)}} \;\leq\; N^{-4/6 + C \epsilon} \,, \qquad
\abs{\Delta x_0^{(3)}} \;\leq\; N^{-7/6 + C \epsilon} \,.
\end{align}
From \eqref{dx0} and \eqref{defY}, we find that $\Delta x_0^{(3)}$ is a finite sum of terms of the form
\be\label{366}
\ind{t = 0} \,\sum_{k\neq \al,\beta,a,b} \frac {\eta \, \overline{m_{sc}^2}} {\pi N^{1/2}} \, R_{\alpha k} \, 
\overline{R_{\beta a} V_{ab} V_{ba} V_{ab} R_{bk}}
+ O(N^{-3/2+C\e})
\ee
\hp, where the other terms are obtained from \eqref{366} as described after \eqref{defY}.

Similarly, we find
\be
\Delta y_0=\Delta y^{(1)}_0+\Delta y^{(2)}_0+\Delta y^{(3)}_0+ O(N^{-7/3+C\e})
\ee
and
\begin{align} \label{Delta y estimates}
\abs{\Delta y_0^{(1)}} \;\leq\; N^{-5/6 + C \epsilon} \,, \qquad
\abs{\Delta y_0^{(2)}} \;\leq\; N^{-8/6 + C \epsilon} \,, \qquad
\abs{\Delta y_0^{(3)}} \;\leq\; N^{-11/6 + C \epsilon} \,.
\end{align}

Now we insert these bounds into \eqref{356}. Recall that the upper index $l$ in $\Delta x_0^{(l)}$ and $\Delta 
y_0^{(l)}$ counts the number of $V$-matrix elements. Thus we find, recalling \eqref{356} and the power counting 
estimates \eqref{Delta x estimates} and \eqref{Delta y estimates}, that there is a random variable $A_5$, depending on 
the randomness only through $Q$ and the two first moments of $h_{ab}$, such that
\begin{multline}\label{367}
\E_{ab}\,\theta \left[  \int_{I }  x^S\, q(  y^S)  \, \dd E \right]
-\E_{ab}\,\theta \left[   \int_{I }  x^R\, q(  y^R)  \, \dd E \right] \\
=\; \E_{ab}\,
\theta' \left[   \int_{I }  x^R\, q(  y^R)  \, \dd E \right]
 \left[  \int_{I }\Big( \Delta x_0^{(3)} \, q ( y^R) + x^Rq '( y^R) \, \Delta y_0^{(3)} \, \dd E \right]
\\
+ A_4 + A_5 +o(N^{-2 + t})\,.
\end{multline}
\hp. Moreover, by the same power counting estimates we find that the second line of \eqref{367} is bounded by 
$o(N^{-1})$.  We use this rough bound in the case $a = b$, and get
\begin{multline}
\E_{ab} \,\theta \left[  \int_{I }  x^S\, q(  y^S)  \, \dd E \right]
-\E_{ab} \, \theta \left[   \int_{I }  x^R\, q(  y^R)  \, \dd E \right] \\
=\; \ind{a \neq b} \, \E_{ab} \,
\theta' \left[   \int_{I }  x^R\, q(  y^R)  \, \dd E \right]
 \left[  \int_{I }\Big( \Delta x_0^{(3)} \, q ( y^R) + x^Rq '( y^R) \, \Delta y_0^{(3)} \, \dd E \right]
\\
+ A_4 + A_5 +o(N^{-2 + t}) +o(N^{-2 +{\bf 1}(a=b)})
\end{multline}
\hp.

Hence Lemma \ref{lemma: main comparison estimate} is proved if we can show that, for $a\neq b $, we have
\be\label{368}
\E\,\theta' \left[   \int_{I }  x^R\, q(  y^R)  \, \dd E \right] \left[  \int_{I }\Big(
 \Delta x_0^{(3)} \, q ( y^R) + x^Rq '( y^R) \Delta y_0^{(3)} \, \Big)\dd E \right] \;=\; o(N^{-2})
\ee
\hp.
This is proved below.
\end{proof}

\begin{proof}[Proof of \eqref{368}]
We shall prove, for $a \neq b$, that
\be\label{369}
\E\,\theta' \left[   \int_{I }  x^R\, q(  y^R)  \, \dd E \right]
 \left[  \int_{I } \Delta x_0^{(3)} \, q ( y^R) \, \dd E \right] \;=\; o(N^{-2})\,.
\ee
The other term on the left-hand side of \eqref{368} is estimated similarly. Let us abbreviate
\begin{align} \label{def of BR}
B^R \;\deq\; \theta' \left[   \int_{I }  x^R\, q(  y^R)  \, \dd E \right]\,.
\end{align}
From \eqref{estimate on x(E)} and the assumption on $\theta$, we find that $\abs{B^R} \leq N^{C \epsilon}$ \hp.

We shall estimate the contribution to \eqref{369} of one term of the form \eqref{366}. Recalling that $\E_{ab} \, 
\abs{V_{ab}}^3 = O(1)$ and $m_{sc}=O(1)$, we find the bound
\begin{multline} \label{estimate of a repr term}
N^{-1/6 + C \epsilon} \, \max_{k\neq \al,\beta,a,b} \absbb{\E \, B^R \int_I R_{\alpha k} \, \overline{R_{\beta a} 
R_{bk}} \, q(y^R) \, \dd E}
+ o(N^{-2})
\\
\leq\;
N^{-5/6 + C \epsilon} \, \max_{k\neq \al,\beta,a,b} \sup_{E \in I} \absbb{\E \, B^R R_{\alpha k} \, \overline{R_{\beta 
a} R_{bk}} \, q(y^R)}
+ o(N^{-2})\,.
\end{multline}
The proof of \eqref{368} is therefore complete if we can show that, assuming the sets $\{\alpha, \beta\}, \{a\}, \{b\}, 
\{k\}$ are disjoint, we have
\be\label{373}
\left|\E \, R_{\al k}\overline{(R_{\beta a} R_{bk})} \, B^R  q ( y^R)
 \right| \;\leq\; N^{-4/3+C\e}\,.
\ee
In order to prove \eqref{373}, we first use a simple resolvent expansion to show that \hp
\begin{align} \label{374}
\absB{R_{\al k}\overline{(R_{\beta a} R_{bk})} \, B^R q(y^R) - S_{\al k}\overline{(S_{\beta a} S_{bk})} \, 
B^S q(y^S)} \;\leq\; N^{-4/3 + C \epsilon}\,,
\end{align}
where $B^S$ is defined analogously to \eqref{def of BR} with $R$ replaced by $S$. Therefore it suffices to prove
\be\label{375}
\left|\E\, S_{\al k}\overline{(S_{\beta a} S_{bk})} \, B^S  q ( y^S)
 \right|\leq N^{-4/3+C\e}
\ee

In order to complete the proof, we introduce some notation. Recall that $H \equiv H_{\gamma - 1}$ and $S = (H - 
z)^{-1}$.  We define $H^{(a)}$ as the matrix obtained from $H$ by setting its $a$-th column and $a$-th row to be zero.  
For any function $F \equiv F(H)$ we define $F^{(a)} \deq F(H^{(a)})$.
We now remove the $a$-th row and column from $H$ in \eqref{375}, which we can do with a negligible error. The key 
identity is the following resolvent identity, proved in Lemma 4.2 of \cite{EYY}: For $k \neq i,j$ we have
\begin{align} \label{Gij Gijk}
S_{ij} \;=\; S_{ij}^{(k)} + \frac{S_{ik} S_{kj}}{S_{kk}}\,.
\end{align}
Using \eqref{Gij Gijk}, one readily sees that
\be\label{378}
\left|
  S_{\al k}\overline{(S_{\beta a} S_{bk})} B^S q ( y^S)-
S^{(a)}_{\al k}\overline{S_{\beta a} S^{(a)}_{b k}} \,
  \theta' \left[   \int_{I } ( x^S)^{(a)}\, q(  (y^S)^{(a)})  \, \dd E' \right]
q \pb{( y^S)^{(a)}}
  \right| \;\leq\;  N^{-4/3+C\e}\,.
\ee
Moreover, we have
\begin{multline} \label{867}
S^{(a)}_{\al k}\overline{S_{\beta a} S^{(a)}_{b k}} \,
  \theta' \left[   \int_{I } ( x^S)^{(a)}\, q(  (y^S)^{(a)})  \, \dd E' \right]
q \pb{( y^S)^{(a)}}
\\
=\; \pbb{S_{\al k}\overline{S_{b k}} \,
  \theta' \left[   \int_{I } x^S\, q(y^S)  \, \dd E' \right]
q( y^S)}^{(a)} \; \overline{S_{\beta a}}\,.
\end{multline}

Next, we claim that the conditional expectation -- with respect to the variables in the $a$-th column of $H$ -- of 
$S_{\beta a}$ is much smaller than its typical size. To that end, we use the identities, valid for $i \neq j$,
\begin{align} \label{376}
S_{ij} \;=\; - S_{ii} \sum_{k \neq i} h_{ik}  S_{kj}^{(i)}\,,\qquad S_{ij} \;=\; - S_{jj} \sum_{k \neq j}  
S_{ik}^{(j)}h_{kj}\,,
\end{align}
proved in \cite{EKYY2}, Lemma 6.10. Now using \eqref{376} we find
\begin{align} \label{splitting of S}
-S_{\beta a} \;=\; \sum_{j \neq a} S_{aa} S_{\beta j}^{(a)} h_{ja} \;=\; \sum_{j \neq a} m_{sc} S_{\beta j}^{(a)} h_{ja}
+ (S_{aa} - m_{sc}) \sum_{j \neq a}  S_{\beta j}^{(a)} h_{ja}\,.
\end{align}
The conditional expectation with respect to the variables in the $a$-th column of $H$ applied to the first term on the 
right-hand side of \eqref{splitting of S} vanishes; hence its contribution to the expectation of \eqref{867} also 
vanishes. In order to estimate the second term on the right-hand side of \eqref{867}, we note that \hp
\begin{align*}
\abs{S_{aa} - m_{sc}} \;\leq\; N^{-1/3 + C \epsilon}\,,
\end{align*}
by Lemma \ref{lemma: Rij bound}. Moreover, using the large deviation bound (3.9) in \cite{EYYrigi}, we get \hp
\begin{align*}
\absBB{\sum_{j \neq a} S_{\beta j}^{(a)} h_{ja}} \;\leq\; N^{-1/2 + \epsilon} \pBB{\sum_{j \neq a} \abs{S_{\beta 
j}^{(a)}}^2}^{1/2} \;\leq\; N^{-1/2 + \epsilon} \absb{S_{\beta \beta}^{(a)}} + N^\epsilon \max_{j \neq a,\beta} 
\abs{S_{\beta j}^{(a)}} \;\leq\; N^{-1/3 + C \epsilon}\,,
\end{align*}
where in the last step we used \eqref{Gij Gijk} and Lemma \ref{lemma: Rij bound}.  Putting everything together, we find 
that the expectation of \eqref{867} is bounded in absolute value by $N^{-4/3 + C \epsilon}$. By \eqref{378}, this 
completes the proof of \eqref{375}, and hence of \eqref{368}.
\end{proof}

\section{Extension to eigenvalues and several arguments} \label{sect:general edge}
In this section we describe how the arguments of Section \ref{section: proof} extend to general functions $\theta$ as in 
\eqref{11}.

Consider first the case of a single eigenvalue, $\lambda_\beta$, in which case the claim reads
 \be\label{11gen}
 \lim_{N\to \infty} \qb{\E^{\f v} - \E^{\f w}}  \theta \pB{N^{2/3}(\lambda_{\beta} - \gamma_{\beta})} \;=\; 0\,,
\ee
uniformly in $\beta \leq \varphi^\rho$. Denote by $\varrho^{\f v}$ and $\varrho^{\f w}$ the laws of $\lambda_\beta$ in 
the ensembles $H^{\f v}$ and $H^{\f w}$ respectively. Using Theorem \ref{7.1} we find
\be
\E^{\f u} \, \theta (N^{2/3}(\lambda_{\beta} - \gamma_{\beta}) ) \;=\; \int_I \theta(N^{2/3} (E - \gamma_{\beta})) \, 
\varrho^{\f u}(\dd E)
+O(\me^{-\varphi^c})\,,
\ee
where $\f u$ stands for either $\f v$ or $\f w$, and $I$ was defined in \eqref{defEI}. Now integration by parts yields
\be \label{eigenvalue estimate}
\qb{\E^{\f v} - \E^{\f w}}  \theta (N^{2/3}(\lambda_{\beta} - \gamma_{\beta}) )
\;=\;
- \qb{\E^{\f v} - \E^{\f w}} \int_{I}N^{2/3} \, \theta'(N^{2/3}(E - \gamma_{\beta}) ) \, \ind{\lambda_\beta \leq E} \, 
\dd E + O(\me^{-\varphi^c})\,,
\ee
where the boundary terms are of order $O(\me^{-\varphi^c})$ by Theorem \ref{7.1}. Next, we choose a smooth nondecreasing 
function $r_\beta$ that vanishes on the interval $(-\infty, \beta - 2/3]$ and is equal to 1 on the interval $[\beta - 
1/3, \infty)$.  Recalling the definition \eqref{defEL}, we get from \eqref{eigenvalue estimate}
\begin{align*}
\qb{\E^{\f v} - \E^{\f w}}  \theta (N^{2/3}(\lambda_{\beta} - \gamma_{\beta}) )
&\;=\;
- \qb{\E^{\f v} - \E^{\f w}} \int_{I}N^{2/3} \, \theta'(N^{2/3}(E - \gamma_{\beta}) ) \, r_\beta\pb{\cal N(E_L, E)} \, 
\dd E + O(\me^{-\varphi^c})
\\
&\;=\;
- \qb{\E^{\f v} - \E^{\f w}} \int_{I}N^{2/3} \, \theta'(N^{2/3}(E - \gamma_{\beta}) ) \, r_\beta\pb{\tr (\f 1_{[E_L, 
E]} * \theta_{\tilde \eta})(H)} \, \dd E + O\pb{\varphi^C N^{-\epsilon}}\,,
\end{align*}
where in the second step we used the assumption on $\theta$, that $r_\beta'$ is bounded, and Lemma \ref{lem:21} with 
$\tilde \eta \deq N^{-2/3 - 6 \epsilon}$. More precisely, we apply Lemma \ref{lem:21} to estimate, \hp,
\begin{align}
&\mspace{-60mu}
\varphi^C N^{2/3} \int_I \dd E \, \absB{\tr (\f 1_{[E_L, E]} * \theta_{\tilde \eta})(H) - \cal N(E_L, E)}
\notag \\
&\;\leq\; \varphi^C N^{2/3} \int_I \dd E \pB{N^{- \epsilon} + \cal N\pb{E - N^{-2/3 - \epsilon}, E + N^{-2/3 -
\epsilon}}}
\notag \\
&\;\leq\; \varphi^C N^{-\epsilon} + \varphi^C N^{2/3} \sum_{\alpha = 1}^{\varphi^C} \int_I \dd E \, \indb{\abs{E - 
\lambda_\alpha} \leq N^{-2/3 - \epsilon}}
\notag \\ \label{edge error estimate}
&\;\leq\; \varphi^C N^{-\epsilon}\,,
\end{align}
where the first step follows from \eqref{6.10} and the second from Theorem \ref{7.1}.

Integrating by parts again, we find \hp
\begin{multline*}
\qb{\E^{\f v} - \E^{\f w}}  \theta (N^{2/3}(\lambda_{\beta} - \gamma_{\beta}) )
\\
=\;
\qb{\E^{\f v} - \E^{\f w}} \int_{I} \theta(N^{2/3}(E - \gamma_{\beta}) ) \, r_\beta'\pb{\tr (\f 1_{[E_L, E]} * 
\theta_{\tilde \eta})(H)} \, N \im m(E + \ii \tilde \eta) \, \dd E + O\pb{\varphi^C N^{- \epsilon}}\,.
\end{multline*}
Now we may apply the Green function comparison method from Section \ref{pf of 33}. In fact, in this case the analysis is 
easier as we have no fixed indices $i$ and $j$ to keep track of.

The general case, $\theta$ as in \eqref{11}, is treated similarly. Repeating successively the above procedure for each 
argument $\lambda_{\beta_1}, \dots, \lambda_{\beta_k}$, we find that there is a constant $C_k$, depending on $k$, such
that
\begin{multline} \label{general starting point for G comp}
\qb{\E^{\f v} - \E^{\f w}}  \theta \pB{N^{2/3}(\lambda_{\beta_1} - \gamma_{\beta_1}), \ldots, 
N^{2/3}(\lambda_{\beta_k} - \gamma_{\beta_k});  N  \bar u_{\alpha_1} (i_1)  u_{\alpha_1} (j_1),    \ldots, N  \bar 
u_{\alpha_k} (i_k)  u_{\alpha_k} (j_k)}
\\
=\; \qb{\E^{\f v} - \E^{\f w}} \int_{I^k} \dd E_1 \cdots \dd E_k \, \theta \pB{N^{2/3} (E_1 - \gamma_{\beta_1}), \dots, 
N^{2/3} (E_k - \gamma_{\beta_k}); \zeta_1, \dots, \zeta_k}
\\
{} \times{} \prod_{l = 1}^k \qBB{r_{\beta_l}'\pb{\tr (\f 1_{[E_L, E_l]} * \theta_{\tilde \eta})(H)} \, N \im m(E_l + \ii 
\tilde \eta)} + O\pb{\varphi^{C_k} N^{-\epsilon}}\,,
\end{multline}
where we introduced the shorthand
\begin{equation*}
\zeta_l \;\deq\; \frac{N}{\pi} \int_I \dd \wt E \, \wt G_{i_l j_l}(\wt E + \ii \eta) \, q_{\alpha_l}\qb{\tr \f 1_{[E_L, 
\wt E^-]} * \theta_{\tilde \eta}(H)}\,,
\end{equation*}
and set $\eta \deq N^{-2/3 - \epsilon}$; $q_\alpha$ is the function from Lemma \ref{GCC}. Here at each step we used the 
assumption on $\theta$, that $r'_\beta$ is bounded, and the estimate
\begin{equation} \label{integral of der in error}
\int_I \dd E \, N \im m(E + \ii \tilde \eta) \;\leq\; \cal N\pb{-\infty, E + N^{-2/3 - \epsilon /10}} + N^{-\epsilon / 10} 
\;\leq\; N n_{sc}(-2 + \varphi^C N^{-2/3}) + 1 \;\leq\; \varphi^C\,,
\end{equation}
where in the first step we used \eqref{41new}, in the second Theorem \ref{7.1}, and in the third the definition 
\eqref{nsc} of $n_{sc}$.

The randomness on the right-hand side of \eqref{general starting point for G comp} is expressed entirely in terms of 
Green functions; hence \eqref{general starting point for G comp} is amenable to the Green function comparison method 
from Section \ref{pf of 33}.  The complications are merely notational, as we now have $2k$ fixed indices $i_1,j_1, 
\dots, i_k,j_k$ instead of just the two $i,j$.

\section{Eigenvectors in the bulk: proof of Theorem \ref{t2}} \label{section: bulk}

In this section we prove Theorem \ref{t2}. In the bulk the eigenvalue spacing is of order $N^{-1}$ as opposed to 
$N^{-2/3}$ at the edge. Thus, we shall have to take spectral windows of size $\eta = N^{-1 - \epsilon}$. To that end, we 
begin by extending the strong local semicircle law from Theorem \ref{45-1} to arbitrarily small values of $\eta > 0$. 
Recall the notation $z = E + \ii \eta$.

\begin{lemma}\label{Letasc} For any $|E|\leq 5$ and $0 < \eta \leq 10$, we have \hp
\be\label{res43}
\max_{ij}|G_{ij}(z)-\delta_{ij}m_{sc}(z)| \;\leq\; \varphi^C\left( \sqrt{\frac{\im m_{sc}(z)  
}{N\eta}}+\frac{1}{N\eta}\right)
\ee
for large enough $N$.  \end{lemma}
\begin{proof}
By Theorem \ref{45-1}, we only need to consider $\eta \leq y \deq \varphi^{C_1} N^{-1}$ for some $C_1 > 0$. We use the 
trivial bound
\begin{equation*}
\im G_{ii}(E + \ii \eta) \;\leq\; \frac{y}{\eta} \im G_{ii}(E + \ii y) \for \eta \leq y\,,
\end{equation*}
as well as
\begin{equation*}
\absb{G_{ij}(E + \ii \eta)} \;\leq\; C \log N \max_k \im G_{kk}(E + \ii \eta)\,,
\end{equation*}
which follows by a simple dyadic decomposition; see \cite{EYY}, Equation (4.9). Thus we get
\begin{equation*}
\absb{G_{ij}(E + \ii \eta)} \;\leq\; C \log N \, \frac{y}{\eta} \max_k \im G_{kk}(E + \ii y) \;\leq\; \varphi^C 
\frac{y}{\eta} \;\leq\; \frac{\varphi^C}{N \eta}\,.
\end{equation*}
This completes the proof.
\end{proof}

The strategy behind the proof of Theorem \ref{t2} is very similar to that of Theorem \ref{t1}, given in Section 
\ref{section: proof}.
In a first step, we express the eigenvector components using integrals involving resolvent matrix elements $G_{ij}$; in 
a second step, we replace the sharp indicator functions in the integrand by smoothed out functions which depend only on 
the resolvent; in a third step, we use the Green function comparison method to complete the proof.

For ease of presentation, we shall give the proof for the case $\theta = \theta (N \bar u_{\alpha}(i) u_\alpha(j))$; we 
show that
\be\label{12sim}
 \lim_{N\to \infty} \big [ \, \E^{\f v} - \E^{{\f w}}  \,  \big ]  \theta (     N  \bar u_{\alpha } (i )  u_{\alpha } (j 
) ) \;=\; 0\,,
\ee
where $\rho N \leq \alpha \leq (1 - \rho) N$.
As outlined in Section \ref{sect:general edge}, the extension to general functions $\theta$, as given in \eqref{12}, is
an easy extension which we sketch briefly at the end of this section.

We now spell out the three steps mentioned above.

\smallskip
\noindent
{\bf Step 1.}
The analogue of Lemma \ref{CTG} in the bulk is the following result whose proof uses \eqref{lrb} and Lemma \ref{lemlb}, 
and is very similar to the proof of Lemma \ref{CTG} (in fact somewhat easier). We omit further details.

\begin{lemma}\label{CTGgen}
Under the assumption of Theorem \ref{t2}, for any $\epsilon > 0$
there exist constants $C_1$, $C_2$ such that  for $\eta = N^{-1 - \epsilon}$ we have
\be\label{114g}\;\;\;\; \lim_{N\to \infty }
 \max_{\rho N \leq \alpha \leq (1 - \rho) N} \max_{i,j}
\hbb{
\E^{\f u} \, \theta \left( N \bar u_\alpha(i) u_\alpha (j) \right)
 -
\E^{\f u} \,\theta \left[\frac {N}\pi \int_{I_\alpha} \wt G_{ij}(E + \ii \eta) \, \ind{\lambda_{\alpha - 1} \leq E^- \leq 
\lambda_\alpha} \, \dd E \right]} \;=\; 0\,,
\ee
where
\be\label{defEIg}
E^\pm \;\deq\; E\pm(\varphi_N)^{C_1}\eta,\;\;\;\;\;\;  I_\alpha \;\deq\; \qB{\gamma_\al -N^{-1} (\varphi_N)^{C_2} \,,\,  
\gamma_\al+N^{-1}  (\varphi_N)^{C_2}}
\ee
and we introduce the convention $\lambda_0 = - \infty$. Here $\f u$ stands for either $\f v$ or $\f w$.
\end{lemma}

\smallskip
\noindent
{\bf Step 2.}
We choose $\eta = N^{-1-\e}$ for some small enough $\e>0$ and express the indicator function in
\be \label{278}
\E^{\f u} \,\theta \left[\frac {N }\pi \int_{I} \wt G_{ij}(E + \ii \eta ) \, \ind{\lambda_{\alpha - 1} \leq E^- \leq 
\lambda_\alpha}\, \dd E \right]
\ee
using Green functions (as before, we write $I_\alpha \equiv I$).
Using Theorem \ref{7.1}, we know that
\be \label{416}
\eqref{278} \;=\;
\E^{\f u} \,\theta \left[\frac {N }\pi \int_{I} \wt G_{ij}(E + \ii \eta ) \, \indb{\mathcal N(E_L, E^-)=\al-1}\, \dd E 
\right] + o(1)\,,
\ee
where $E_L \deq -2-\varphi^CN^{-2/3}$.

As explained in Section \ref{sect: outline of proof}, the approach in Step 2 has to be modified slightly from the one 
employed in Section \ref{section: proof}. The reason is that the size of the interval $[E_L, E^-]$ is no longer small, 
but of order one.

For any $E_1, E_2 \in [-3,  3]$ and $\eta_d>0$ we define $f(\lambda) \equiv f_{E_1,E_2,\eta_d}(\lambda)$
to be the characteristic function of $[E_1, E_2]$ smoothed on scale $\eta_d$; i.e.\
$f = 1$ on $[E_1, E_2]$, $f = 0$ on $\R\setminus [E_1-\eta_d, E_2+\eta_d]$
and $|f'|\le C\eta_d^{-1}$, $|f''|\le C\eta_d^{-2}$. Let $q \equiv q_\alpha:\R \to\R_+$ be a smooth cutoff function 
concentrated around $\alpha - 1$, satisfying
\[
q(x) = 1 \quad  \text{if} \quad |x - \alpha + 1| \le 1/3,   \qquad q(x) = 0   \quad  \text{if} \quad |x - \alpha + 1| 
\ge 2/3\,.
\]
Now we choose $\eta_d \deq N^{-1-d\e}$, for some fixed $d>2$. Then, using Lemma \ref{Letasc} and an argument similar to 
the proof of Lemma \ref{GCC}, we find that
\be \eqref{416} \;=\; \E^{\f u} \,\theta \left[\frac {N  }\pi \int_{I} \wt G_{ij}(E + \ii \eta) \, q \left(\tr f_{E_L, 
E^-,   \eta_d}(H)\right)
\, \dd E \right]+o(1)\,.
\ee
To simplify notation, we follow the conventions of Section \ref{section: proof} in writing $I \equiv I_\alpha$, $q 
\equiv q_\alpha$ and $f_E \equiv f_{E_L, E^-, \eta_d}$, and set $\alpha = i$ and $\beta = j$. In this notation, we need 
to estimate
\be\label{418g}
\qb{\E^{\f v} - \E^{\f w}} \,\theta \left[\frac {N}\pi \int_{I} \wt G_{\al\beta}(E + \ii \eta) \, q \left(\tr 
f_E(H)\right)
\, \dd E \right]\,.
\ee
Now we express $f_E(H)$ in terms of Green functions using Helffer-Sj\"ostrand functional calculus. Let $\chi(y)$ be a 
smooth cutoff function with support in $[-1, 1]$, with $\chi(y)=1$ for
$|y| \leq 1/2$ and with bounded derivatives. Then we have (see e.g.\ Equation (B.12) of \cite{ERSY})
\be
f_E(\lambda) \;=\; \frac1{2\pi}\int_{\R^2}\frac{\ii \sigma f_E''(e)\chi (\sigma)+ \ii f_E(e) \chi'(\sigma)-\sigma 
f_E'(e)\chi'(\sigma)}{\lambda-e-\ii \sigma} \, \dd e \, \dd \sigma\,.
\ee
Thus we get
\begin{align}
\tr f_E(H) &\;=\; \frac{N}{2\pi}\int_{\R^2}\pB{\ii \sigma f_E''(e)\chi (\sigma)+ \ii f_E(e) \chi'(\sigma)-\sigma 
f_E'(e)\chi'(\sigma)} m(e + \ii \sigma) \, \dd e \, \dd \sigma
\notag \\
&\;=\; \frac{N}{2\pi}\int_{\R^2}\pB{\ii f_E(e) \chi'(\sigma)-\sigma f_E'(e)\chi'(\sigma)} m(e + \ii \sigma) \, \dd e \, 
\dd \sigma
\notag \\ \label{HS split 1}
&\qquad {}+{} \frac{\ii N}{2 \pi} \int_{\abs{\sigma} > \tilde \eta_d} \dd \sigma \, \chi(\sigma) \int \dd e \; f''_E(e) 
\, \sigma  m(e + \ii \sigma)
+
\frac{\ii N}{2 \pi} \int_{-\tilde \eta_d}^{\tilde \eta_d} \dd \sigma \int \dd e \; f''_E(e) \, \sigma  m(e + \ii \sigma)
\,,
\end{align}
where we introduced the parameter $\tilde \eta_d \deq N^{-1 - (d + 1)\epsilon}$. We shall treat the last term of 
\eqref{HS split 1} as an error term. From \eqref{res43} we find \hp
\begin{equation*}
\sigma  m(e + \ii \sigma) \;\leq\; \frac{\varphi^C}{N}\,.
\end{equation*}
Therefore the third term of \eqref{HS split 1} is bounded, \hp, by
\begin{equation} \label{small sigma error}
\absBB{\frac{\ii N}{\pi}\int_{-\tilde \eta_d}^{\tilde \eta_d} \dd \sigma \int \dd e \; f''_E(e) \, \sigma  m(e + \ii 
\sigma)} \;\leq\; \varphi^C \tilde \eta_d \eta_d^{-1} \;=\; \varphi^C N^{-\epsilon}\,,
\end{equation}
where in the first step we used that $\int |f''_E(e)| \, \dd e = O(\eta_d^{-1})$.

\smallskip
\noindent
{\bf Step 3.}
We estimate \eqref{418g} using a Green function comparison argument, similarly to Section \ref{pf of 33}.
As in Section \ref{pf of 33}, we use the notation
\be
x(E) \;=\; \frac {N\eta}\pi \sum_{k\neq \al, \beta}  G_{\al k}(E + \ii \eta)\overline{G_{\beta k}(E + \ii \eta)}\,.
\ee
Similarly to Lemma \ref{dropping diagonal terms}, we begin by dropping the diagonal terms. Using Lemma \ref{Letasc} we 
find
\be
\int_I \left|   \frac{N}{\pi} \wt G_{\al\beta}(E + \ii \eta) - x(E) \right|\, \dd E \;\leq\; \varphi^C N\eta^2 \;\leq\; 
N^{-1 +C\e}
\ee
\hp, so that it suffices to prove
\be
\qb{\E^{\f v}-\E^{\f w}} \,\theta \left[ \int_I x(E) \, q \left(\tr f_E(H)\right)
\, \dd E \right] \;=\; o(1)\,.
 \ee

Using \eqref{small sigma error} we find that it suffices to prove
\be \label{ee427}
\qb{\E^{\f v} - \E^{\f w}} \,\theta \left[  \int_{I} x(E) \, q \pb{y(E)+\wt y(E)}   \, \dd E  \right] \;=\; o(1)\,,
\ee
where
\begin{align}
\label{defy428}
y(E) &\;\deq\; \frac N{2\pi}\int_{\R^2}
\ii \sigma f_E''(e) \chi (\sigma)  \,  m(e + \ii \sigma) \,  {\bf 1} (\abs{\sigma} \geq \tilde \eta_d) \,\dd e \, \dd 
\sigma \,,
\\
\label{427y}
 \wt y(E) &\;\deq\;
 \frac N{2\pi}\int_{\R^2}
\Big( \ii f_E(e) \chi'(\sigma)- \sigma f_E'(e)\chi'(\sigma)\Big)m(e + \ii \sigma) \, \dd e \, \dd \sigma\,.
\end{align}
By a telescopic expansion similar to \eqref{tel}, we find that \eqref{ee427} follows if we can prove, \hp,
\begin{equation}\label{340sws}
\E \, \theta \left[  \int_{I }  x^S(E) \, q\pb{ (y+\wt y)^S(E)} \,\dd E\right]
- \E \, \theta \left[   \int_{I }  x^R(E) \, q\pb{ (y+\wt y)^R(E)} \, \dd E\right]
\;=\; A + o(N^{\ind{a = b} - 2})\,,
\end{equation}
where we use the notation of Section \ref{pf of 33}; here $A$ is a random variable that depends on the randomness only 
through $Q$ and the first four moments of $h_{ab}$ if $a \neq b$, and the first two moments of $h_{ab}$ if $a = b$. (As 
in Section \ref{pf of 33}, $\E$ denotes expectation with respect to the product measure of the $\f v$ and $\f w$ 
ensembles.)

Now we prove \eqref{340sws}. We use the resolvent expansion
\be\label{SR-Nf}
	 S =  R -  N^{-1/2} RVR+  N^{-1} (RV)^2R -  N^{-3/2} (RV)^3R+  N^{-2} (RV)^4R - N^{-5/2} (RV)^5 S.
\ee
Similarly to Section \ref{pf of 33}, we decompose
\begin{align*}
\Delta m \;\deq\; m^S - m^R \;=\; \Delta m_0 + \Delta m_1\,,
\end{align*}
where
\begin{align*}
\Delta m_r \;\deq\; \frac{1}{N} \sum_i (S_{ii} - R_{ii}) \, \indb{r = \ind{i \in \{a,b\}}}\,.
\end{align*}
Using \eqref{SR-Nf} we can expand $\Delta m_r$, for $\abs{\sigma} \geq \tilde \eta_d$ and \hp,
\be
\Delta m_r(e + \ii \sigma) \;=\; \sum_{p=1}^4 \Delta m_r^{(p)}(e + \ii \sigma) + O\pb{N^{-5/2 + C \epsilon} 
\Lambda_\sigma^{2-2r}N^{-r}}\,,
\ee
where
\be
\abs{\Delta m_r^{(p)}} \;\leq\; N^{-p/2 + C \epsilon} \Lambda_\sigma^{2-2r}N^{-r}
\ee
\hp,
and $\Delta m_r^{(p)}$ is a polynomial in the matrix elements of $R$ and $V$, each term containing precisely $p$ matrix 
elements of $V$; here we set $\Lambda_\sigma \deq \sup_{\abs{e} \leq 5} \max_{i \neq j} \abs{G_{ij}(e + \ii \sigma)}$.  
Putting both cases $r = 1,2$ together, we get, for $\abs{\sigma \geq \tilde \eta_d}$ and \hp,
\be\label{434}
\Delta m \;=\; \sum_{p=1}^4 \Delta m^{(p)}+O\pb{N^{-5/2 + C \epsilon} (\Lambda_\sigma^{2} + N^{-1})}
\,,\qquad \abs{\Delta m^{(p)}} \;\leq\; N^{-p/2 + C \epsilon} (\Lambda_\sigma^{2} + N^{-1})\,.
\ee

We may now estimate the variables $x, y$, and $\wt y$.
Let us first consider the variables $\wt y$. From the definition of $\chi$, we find that in the integrand of 
\eqref{427y} we have $\sigma \geq c$ and therefore by Theorem \ref{45-1} we have $\Lambda_\sigma \leq \varphi^C 
N^{-1/2}$ \hp.
Thus we get from \eqref{427y}
 \be\label{dwty}
\Delta \wt y(E) \;=\; \sum_{p=1}^4 \Delta \wt y^{(p)}(E) +O(N^{-5/2 + C \epsilon})\,, \qquad \abs{\Delta \wt y^{(p)}(E)} 
\;\leq\; N^{-p/2 + C \epsilon}
\ee
\hp.

In order to estimate the contributions of the variables $y$, we integrate by parts, first in $e$ and then in $\sigma$, 
to obtain
\begin{multline} \label{IBP for HS}
\absBB{N \int_{\R^2}
\sigma f_E''(e) \chi (\sigma) \, \Delta m^{(p)}(e + \ii \sigma) \,  {\bf 1} (\abs{\sigma} \geq \tilde \eta_d) \,\dd e \, 
\dd \sigma}
\\
\leq\; CN \absbb{ \int \dd e \, f_E'(e) \, \tilde \eta_d \, \Delta m^{(p)}(e + \ii \tilde \eta_d)}
+ CN \absbb{\int \dd e f'_E(e) \int_{\tilde \eta_d}^\infty \dd \sigma \, \chi'(\sigma) \sigma \, \Delta m^{(p)}(e + \ii 
\sigma)}
\\
+ CN \absbb{\int \dd e f'_E(e) \int_{\tilde \eta_d}^\infty \dd \sigma \, \chi(\sigma) \, \Delta m^{(p)}(e + \ii 
\sigma)}\,.
\end{multline}
Using \eqref{434}, it is easy to see that the sum of the two first terms of \eqref{IBP for HS} is bounded by $N^{-p/2 + 
C \epsilon}$. In order to estimate the last term of \eqref{IBP for HS}, we use \eqref{434} and \eqref{res43} to get the 
bound
\begin{equation*}
C N \int_{\tilde \eta_d}^1 \dd \sigma \pbb{\frac{1}{N \sigma} + \frac{1}{(N \sigma)^2} + \frac{1}{N}} N^{-p/2 + C 
\epsilon}
\;\leq\; N^{-p/2 + C \epsilon}\,.
\end{equation*}
Thus we find that
\be\label{dy}
\Delta  y(E) \;=\; \sum_{p=1}^4 \Delta y^{(p)}(E)+O(N^{-5/2+C\e})
\,,\qquad \abs{\Delta  y^{(p)}(E)} \;\leq\;  N^{-p/2 + C \epsilon}
\ee
\hp. 

Finally, as in \eqref{346a}, we find that
\be\label{346e}
\Delta x(E) \;=\;  \sum_{p=1}^{4}\Delta x ^{(p)}(E) + O(N^{-3/2+C\e})
\ee
\hp.
Moreover, we have the bound
\be\label{346af}
  \int_{I }   |\Delta x ^{(p)}(E)|
 \, \dd E \;\leq\; N^{ -p/2 +C\e}
 \ee
\hp.
This concludes our estimate of the terms in the resolvent expansion of $x$,$y$, and $\wt y$.

Now using the power counting bounds from \eqref{dwty}, \eqref{dy}, \eqref{346e}, and \eqref{346af}, we may easily 
complete the Green function comparison argument to prove \eqref{340sws}, as in Section \ref{pf of 33}.

Finally, we comment on how to deal with more general observables $\theta$, as in \eqref{12}. The basic strategy is the
same as in Section \ref{sect:general edge}. In fact, the argument is simpler because the errors made in replacing sharp
indicator functions with smooth indicator functions are easier to control in the bulk. In Section \ref{sect:general edge}, the
relatively large errors arising from the soft edge of the function $\f 1_{[E_L, E]} * \theta_{\tilde \eta}$ were
controlled by Lemma \ref{lem:21}. In the bulk, we replace $\f 1_{[E_L,E]}$ with the function $f_E$ whose edges are
sharper. Thus, in the bulk the error resulting from this replacement is bounded by $\cal N (E^- - \eta_d, E^- +
\eta_d)$, whose integral may be estimated exactly as in \eqref{edge error estimate}. The estimate \eqref{integral of der in
error} is replaced by the trivial estimate $\int_I \dd E \, \abs{\partial_E \tr f_E(H)} \leq \varphi^C$.



\end{document}